\documentclass[11pt]{article}

\def\comments{1}
\def\blind{0}

\usepackage{rose-preamble}

\newcommand{\supp}{\text{support}}
\newcommand{\x}{\mathbf{x}}
\newcommand{\xp}{\mathbf{x_0}}
\newcommand{\Tip}{\frac{3n}{kC^3\sqrt{\log C}}}
\newcommand{\Test}{\frac{3n}{8k'C\log C}}
\newcommand{\test}{3n/(8k'C\log C)}
\newcommand{\cx}{c_{\ell}(\mathbf{x})}
\newcommand{\cxo}{c_{\ell_1}(\mathbf{x})}
\newcommand{\cxt}{c_{\ell_2}(\mathbf{x})}
\newcommand{\chat}{\hat{c}_{\ell}(\mathbf{x})}
\newcommand{\Zmax}{Z_{\text{max}}}
\newcommand{\Lap}{{\mathrm{TLap}}}
\newcommand{\Zl}{Z_{\ell}}
\newcommand{\Bl}{B_{\ell}}
\newcommand{\main}{{\tt interior-point-main}~}

\newcommand{\est}{{\tt estimate-first-moment}~}
\newcommand{\estarg}{{\tt estimate-first-moment}($\x$;$\eps$,$\delta$,$C$)~}
\newcommand{\ip}{{\tt find-interior-point}~}
\newcommand{\iparg}{{\tt find-interior-point}($\x$;$\eps$,$\delta$,$C$,$\hat{m}$) ~}
\newcommand{\up}{\left(2k'C\sqrt{\log C}\right)}
\newcommand{\q}{\mathbf{q}}
\newcommand{\csq}{c_{\ell^*}(\q)}
\newcommand{\cshatq}{\hat{c}_{\ell^*}(\q)}
\newcommand{\lb}{kC^3\sqrt{\log C}\left(\log(2/\beta) + \frac{16\ln(16/\delta)}{\eps}\right)}
\newcommand{\df}[1]{\textbf{\emph{#1}}}
\def\ham{\ensuremath{\mathrm{\mathbf{Ham}}}}

\title{Differentially Private Medians and \\ Interior Points for Non-Pathological Data}
\ifnum\blind=1
    \author{Anonymized for Submission}
\else
    \author{
        Maryam Aliakbarpour\thanks{Khoury College of Computer Sciences, Northeastern University and Department of Computer Science, Boston University.  Supported by NSF awards CNS-2120667, CNS-2120603, CCF-1934846, and BU’s Hariri Institute for Computing. \url{m.aliakbarpour@northeastern.edu}} \and 
        Rose Silver\thanks{Khoury College of Computer Sciences, Northeastern University. Supported by NSF awards CCF-1750640 and CNS-2120603.  \url{silver.r@northeastern.edu}} \and 
        Thomas Steinke\thanks{Google. \url{steinke@google.com}} \and 
        Jonathan Ullman\thanks{Khoury College of Computer Sciences, Northeastern University.  Supported by NSF awards CCF-1750640 and CNS-2120603.  \url{jullman@ccs.neu.edu}}
    }
\fi
\date{}

\begin{document}

\maketitle

\begin{abstract}
    We construct differentially private estimators with low sample complexity that estimate the median of an arbitrary distribution over $\R$ satisfying very mild moment conditions.  Our result stands in contrast to the surprising negative result of Bun et al. (FOCS 2015) that showed there is no differentially private estimator with any finite sample complexity that returns any non-trivial approximation to the median of an arbitrary distribution.
\end{abstract}




\section{Introduction}

A statistical estimator is an algorithm that takes data drawn from an unknown distribution as input and tries to learn something about that distribution.  While the input data is only a conduit for learning about the distribution, many statistical estimators also reveal a lot of information that is specific to the input data, which raises concerns about the \emph{privacy} of people who contributed their data.  In response, we can try to design estimators that are \emph{differentially private (DP)}~\cite{DworkMNS06}, which ensure that no attacker can infer much more about any person in the input data than they could have inferred in a hypothetical world where that person's data had never been collected.

Differential privacy is a strong constraint that imposes significant costs even for very simple statistical estimation tasks.  In this paper we focus on two such tasks: \emph{interior point estimation} and \emph{median estimation}.  In the interior point problem, we have a distribution $P$ over $\R$, and our goal is simply to output some point $y$ with
\begin{equation}
    \inf \supp(P) \leq y \leq \sup \supp(P).
    \label{eq:interior_point}
\end{equation}
There is a trivial estimator for solving the interior point problem---draw a single sample from $P$ and output it---but this estimator is clearly not private.  More generally, we can try to find an $\alpha$-approximate median of the distribution, which is a point $y$ such that
\begin{equation}
    \tfrac{1}{2} - \alpha \leq \Pr_{x \gets P}[x \leq y] \leq \tfrac{1}{2} + \alpha.
\end{equation}
There is also a simple estimator for computing an approximate median---draw $O(1/\alpha^2)$ samples and return the median of the samples---but this estimator also fails to be private.  While these problems are nearly trivial to solve without a privacy constraint, a remarkable result of Bun, Nissim, Stemmer, and Vadhan~\cite{BunNSV15} showed that there is no differentially private estimator that takes any finite number of samples and outputs even just an interior point of an arbitrary distribution! Since the interior point problem is a special case of finding an approximate median, learning threshold functions, learning halfspaces, and more, this negative result has far reaching implications.

In light of this negative result, there have been two main approaches to privately solving the interior point problem and its generalizations.  The first is to assume the data comes from a finite domain, such as the integers $[T] := \{1,2,\dots,T\}$, in which case the optimal sample complexity is now known to be $\smash{n=\tilde{O}(\log^* T)}$~\cite{BeimelNS13b,BunNSV15,alon2019private,bun2018composable,kaplan2020privately,cohen2022optimal}.  The primary drawback of this approach is that a lot of data does not live in a discrete domain, so it has to be rounded to a discrete domain. This rounding requires both care and additional knowledge by the user of the algorithm so as to not destroy too much information. However, rounding can be cumbersome, and additional knowledge may be difficult to obtain when the data is private.

The second approach, which is the approach we adopt in this paper, is to assume that the distribution satisfies some additional properties that allow us to bypass the Bun et al.\ lower bound.  Results along this line have considered a range of assumptions such as Gaussian distributions~\cite{KarwaV18}, distributions with positive density in a neighborhood around the median~\cite{DworkL09,tzamos2020optimal,brunel2020propose}, and distributions with bounded density~\cite{haghtalab2020smoothed}, which hold for common parametric families of continuous distributions.  While these assumptions are natural, they are quite restrictive, and there is a huge gap between distributions satisfying these assumptions and the highly contrived distributions constructed in the Bun et al.\ lower bound.

In this work we show that a very weak, and qualitatively different assumption---\emph{bounded normalized variance}---is sufficient to privately find an interior point.  Specifically, a distribution $P$ satisfies $C$-bounded normalized variance if 
\begin{equation}
\frac{\E_{X\leftarrow P}[|X-\mu|^2]}{\E_{X\leftarrow P}[|X-\mu|]^2} \le C.
\label{eq:bsv}
\end{equation}
for some constant $C \geq 1$.  This assumption is satisfied by essentially all natural parametric families of distributions, both discrete and continuous, and thus we consider it quite weak.  
\begin{thm}[Informal]
There is an $(\eps,\delta)$-differentially private algorithm that takes $n = O_{C,\eps,\delta}(1)$ samples from an arbitrary distribution $P$ over $\R$ satisfying $C$-bounded normalized variance \eqref{eq:bsv} and, with high probability, returns an interior point \eqref{eq:interior_point} of $P$.  Here $O_{C,\eps,\delta}(1)$ denotes that the number of samples depends only on the privacy parameters $\eps,\delta>0$ and the bounded normalized variance parameter $C$. 
\end{thm}
Our result shows that the impossibility of privately finding an interior point \cite{BunNSV15} is extremely brittle, and can be circumvented for essentially all reasonable data distributions.

In the worst-case setting, finding an approximate median can actually be reduced to finding an interior point, however this reduction does not preserve the property of bounded normalized variance, so we cannot use it directly to obtain a private median algorithm. Moreover, we will see that we can turn any distribution on a bounded support into a distribution satisfying bounded normalized variance without changing the median, so our assumption is not sufficient to circumvent the lower bound. Nonetheless, we show that a slight (and necessary) strengthening of this assumption is enough to find an approximate median.  Intuitively, this assumption is \emph{bounded normalized variance around the median}, which means that bounded normalized variance holds even if we condition on the part of $P$ that lies between the $\frac12-\alpha$ and $\frac12+\alpha$ quantiles.
\begin{thm}[Informal] \label{thm:main-intro-median}
There is an $(\eps,\delta)$-differentially private algorithm that takes $n = O_{C,\eps,\delta,\alpha}(1)$ samples from an arbitrary distribution $P$ over $\R$ satisfying $C$-bounded normalized variance around the median, and, with high probability, returns an $\alpha$-approximate median of $P$. Here $O_{C,\eps,\delta,\alpha}(1)$ denotes that the number of samples is an absolute constant depending on $C,\eps,\delta,\alpha$.
\end{thm}

We note that there are many other richer tasks, such as privately learning halfspaces and privately finding a point in the convex hull in $\R^d$, where the best private algorithms for worst-case distributions are based on reductions to interior point or closely related problems that can be reduced to interior point~\cite{beimel2019private}.  Theorem~\ref{thm:main-intro-median} suggests that identifying approximate mild distributional assumptions to make these problems tractable is a fruitful direction.

\subsection{Technical Overview}

At a high level, our interior point algorithm follows the approach taken by Karwa and Vadhan~\cite{KarwaV18} for finding an approximate median of a Gaussian distribution, but with a much more general analysis that allows us to rely on only weak assumptions about the distribution.  First, suppose that we know the first central absolute moment of the distribution, and have rescaled the distribution so that
$$
\E_{X\leftarrow P}[|X-\mu|] = 1 \textrm{ and } \E_{X\leftarrow P}[|X-\mu|^2] \le C,
$$
where $\mu = \E[X]$. In this case, by Chebyshev's inequality, we know that most of the probability mass for $P$ is not too far from the mean $\mu$. Moreover, it cannot be that the mass is almost entirely contained in a single sub-interval of size $\le 1/2$, as this (combined with the fact that outliers in $P$ are relatively rare) would imply that $\E_{X \leftarrow P}[|X - \mu|] < 1$. Thus, if we divide the real line into an infinite set of intervals
$$
\dots, [-1, -\tfrac{1}{2}), [-\tfrac{1}{2}, 0), [0,\tfrac{1}{2}), [\tfrac{1}{2}, 1), \dots,
$$
then there will be at least two distinct intervals that contain a significant amount of mass (more than $1/\poly(C)$ mass).  Using standard techniques for computing differentially private histograms, we can identify two of these intervals privately, and any boundary between them must be an interior point of the distribution.

The next step is to resolve the fact that we do not know the first central absolute moment $\E_{X\leftarrow P}[|X-\mu|]$, and we need to privately estimate this quantity up to a small multiplicative factor.  To do this, we take a set of $2n$ samples $x_1,\dots,x_{2n}$ and create a new set of $n$ samples $y_i = |x_{2i-1} - x_{2i}|$.  Note that each of these samples $y_i$ is sampled as $|X - X'|$ where $X,X'$ are drawn independently from $P$.  We will use the $y_i$s to approximate $\E[|X-X'|]$, which is in turn a constant-factor approximation of $\E[|X-\mu|]$.  Specifically, we divide $[0,\infty)$ into the infinite set of intervals
$$
\dots, [\tfrac{1}{8},\tfrac{1}{4}), [\tfrac{1}{4},\tfrac{1}{2}), [\tfrac{1}{2},1), [1,2), [2,4), [4,8), \dots
$$
Using the bounded normalized variance condition, we argue that the largest interval that contains a significant amount of mass is a good approximation to the first central absolute moment (up to a $\poly(C)$ factor). As before, this largest interval can be identified privately via techniques for computing differentially private histograms. Chaining this algorithm with the previous algorithm gives us our complete algorithm for finding an interior point of a distribution with bounded normalized variance.

We might hope that bounded normalized variance is sufficient to also find a good $\alpha$-approximate median for small $\alpha$, but that is actually false. To see why, first note that the Bun et al.~\cite{BunNSV15} lower bound says that there is no differentially private algorithm that can find an interior point of an arbitrary distribution, even if the distribution is supported on some bounded interval, such as $[-\frac12,\frac12)$.  Given an arbitrary distribution $P$ on this interval, we can create a new distribution $P'$ by adding mass at $-1$ with probability $1/4$ and mass at $+1$ with probability $1/4$.  A simple calculation shows that this new distribution will have $O(1)$-bounded normalized variance.  Moreover, any $1/5$-approximate median for $P'$ will be an interior point of $P$.  Since we can easily simulate access to $P'$ using access to $P$, any private algorithm for computing an approximate median of a distribution with bounded normalized variance can be used to privately compute an interior point of an arbitrary distribution, which is ruled out by Bun et al.~\cite{BunNSV15}.

Thus, our algorithm for finding an $\alpha$-approximate median requires a stronger assumption on the distribution $P$.  First, observe that an $\alpha$-approximate median of $P$ is just any interior point of the distribution $P_{\alpha}$ that consists only of the middle $2\alpha$ \emph{slice} of the distribution $P$.  That is, $P_{\alpha}$ is the distribution $P$ restricted to the space between the $\frac12 - \alpha$ and $\frac12 + \alpha$ quantiles of $P$.  We now assume that the distribution $P_{\alpha}$ has bounded normalized variance.  Intuitively, our algorithm works by finding an interior point of the distribution $P_{\alpha}$ by using a subset of our samples from $P$ to simulate samples from $P_{\alpha}$, but we cannot exactly generate samples from $P_{\alpha}$ without knowing the quantiles of the distribution $P$ itself, which is exactly what we are trying to estimate.  To get around this issue, we instead take a collection of $n$ samples, then sort them to obtain $x_{(1)} \leq x_{(2)} \leq \cdots \leq x_{(n)}$, and then use the middle $(1 - 1/k) 2\alpha n$ samples $\xp$ as an approximation of samples from $P_{\alpha}$ (where $k$ is some well-chosen quantity). The $1 - 1/k$ term is to ensure that the samples in $\xp$ come from \emph{within} $P_{\alpha}$, instead of being in a situation where a small number of them may come from outside of it. While this collection $\xp$ of samples does not quite have the same distribution as i.i.d.~samples from $P_{\alpha}$, we are nonetheless able to argue that they come from a distribution with $O(C)$-bounded normalized variance. Thus our interior point algorithm will succeed in identifying an interior point of $P_{\alpha}$, and thus in identifying an $\alpha$-approximate median of $P$.

\section{Preliminaries}

Let $P$ be a data distribution. We indicate that a data point $x$ is drawn from $P$ by writing $x \leftarrow P$. We indicate that $\x = (x_1,\ldots,x_n)$ is a set of $n$ i.i.d. data points drawn from distribution $P$ by writing $\x \gets P^n$. We refer to $\x$ as a \df{dataset}. We use $\mathbb{I}_{\varphi}$ to denote the indicator random variable for the property $\varphi$.

\subsection{Background}
We say that two datasets $\x$ and $\x'$ are \df{neighboring datasets} if they differ in at most one data point, i.e. $\ham(\x,\x')=1$ where $\ham$ denotes the Hamming distance.
\begin{defn}
Let $\eps > 0$, $0 < \delta < 1$. An algorithm $\mathcal{A}:\mathcal{X}^n \rightarrow \mathcal{Y}$ is $(\eps,\delta)$-differentially private if, for every $E \subseteq \mathcal{Y}$ and neighboring datasets $\mathbf{x},\mathbf{x'} \in \mathcal{X}^n$, $\mathcal{A}$ satisfies
$\Pr[\cA(\mathbf{x}) \in E] \le e^{\eps} \Pr[\cA(\mathbf{x'}) \in E] + \delta.$
\label{def:dp}
\end{defn}

\begin{defn}
Let $f: \mathcal{X}^n \rightarrow \mathbb{R}^d$ be a function. The \df{global sensitivity} $\Delta$ of $f$ is defined as $$\Delta \coloneqq \sup_{\mathbf{x},\mathbf{x'} \in \mathcal{X}^n \atop \ham(\x,\x') = 1} ||f(\mathbf{x}) - f(\mathbf{x'})||_{1}.$$
\end{defn}

\paragraph{The truncated Laplace mechanism.} We make use of the standard approach of adding noise proportional to the global sensitivity~\cite{DworkMNS06} to ensure differential privacy.  Since we are interested in adding noise to a histogram with infinitely many bins, we need to make use of the \emph{truncated Laplace distribution} rather than the standard Laplace distribution.  Given parameters $\lambda, \Zmax$, we define the truncated Laplace distribution $\Lap(\lambda,\Zmax)$ over the support $[-\Zmax,\Zmax]$ with density $f(z) \propto e^{-|z|/\lambda}$.


\begin{restatable}{lem}{lemTLap}
    [Truncated Laplace Mechanism] 
    \label{lem:lap}
    Let $f:\mathcal{X}^n\rightarrow \R$ be a function with global sensitivity $\Delta$. For every $\eps, \delta \in (0,1)$, if $\Zmax$ is at least $\Delta \ln(4/\delta)/\eps$, then the truncated Laplace mechanism  $M(\x) \coloneqq f(\x) + Z,$ where $Z \leftarrow \Lap(\frac{\Delta}{\eps},\Zmax)$ is $(\eps,\delta)$-differentially private.
\end{restatable}
While the above lemma is  considered folklore in the field, we include a proof in Section~\ref{sec:TLap} for the sake of completeness. A standard application of this mechanism is computing histograms, possibly in infinite dimension (see e.g.~\cite{Vadhan16}).
\begin{lem}[Differentially Private Histograms]
    Let $\cX$ be a domain and let $\cX_1,\dots,\cX_m$ be a partition of the domain into (a possibly infinite number of) bins.  Define the function $f \from \cX^n \rightarrow \R^m$ as $f(\x)_j = \sum_{i=1}^{n} \mathbb{I}_{x_i \in \cX_j}$.  Then the mechanism
    $$M(\x) \coloneqq f(\x) + (Z_1,\ldots,Z_m),$$ where each $Z_j\leftarrow \Lap\left(\frac{4}{\eps},\frac{8\ln(8/\delta)}{\eps}\right)$ is $(\eps,\delta)$-differentially private.
    \label{lem:private-histogram}
\end{lem}

\subsection{Problem Definitions}

\paragraph{The interior point problem.} Let $P$ be an unknown distribution over $\mathbb{R}$, and let $\mathbf{x} \leftarrow P^n$ be an $n$-dimensional data set whose entries are sampled i.i.d. from $P$. As noted earlier, we define an \df{interior point of $P$} to be any point $y$ satisfying $\inf \supp(P) \leq y \leq \sup \supp(P)$.
Similarly, we define an \df{interior point of $\mathbf{x}$} to be any point $y$ satisfying $\underset{x_i \in \mathbf{x}}{\min}\,x_i \le y \le \underset{x_i \in \mathbf{x}}{\max}\,x_i.$
Observe that any interior point for $\x$ is guaranteed to also be an interior point of $P$.

\paragraph{Approximate medians.} For a distribution $P$, let $F_P(x) \coloneqq \Pr_{X\leftarrow P}[X \le x]$ denote the CDF of distribution $P$. We use $Q_P(p) \coloneqq F_P^{-1}(p)$ to denote the $p$-th quantile of $P$ for any $p \in [0,1]$. That is, $Q_P(p) = \inf \{x \mid F_P(x) \ge p\}$ for $p \in [0,1]$.
For a dataset $\x$, we can similarly define $Q_{\x}(p) \coloneqq x_{(\lfloor pn\rfloor)}$ for $p \in [1/n,1]$. We say that $\hat{m}$ is an $\alpha$-approximation to the median $Q_P(0.5)$ if $|F_P(\hat{m}) - 0.5| \le \alpha$. 

\smallskip

For the interior point problem, we focus on data distributions with the following property:
\begin{defn} 
Let $P$ be a distribution with mean $\mu = \E_{X \leftarrow P}[X]$. The distribution $P$ has \df{$C$-bounded normalized variance} for some value $C$ if
\[
\frac{\E_{X\leftarrow P}[|X-\mu|^2]}{\E_{X\leftarrow P}[|X-\mu|]^2} \le C.
\]
As shorthand, we will sometimes say simply that such a distribution is \df{$C$-bounded}.
\end{defn}

For the approximate median problem, we will be interested in the case where the \emph{middle $2\alpha$-percentile} of the distribution is $C$-bounded, rather than the entire distribution itself.

\section{Interior Point Algorithm}
In this section, we introduce an algorithm (Algorithm~\ref{alg:ip-alg}) that privately solves the interior point problem when the data points are coming from a $C$-bounded distribution. Recall that, given a dataset $\x$, the goal of such an algorithm is to privately output a point $y$ that falls within the minimum and maximum values in $\x$. Formally, we show that Algorithm~\ref{alg:ip-alg} is an $(\eps,\delta)$-differentially private algorithm that, if $P$ is $C$-bounded, returns an interior point of $\x$ with probability at least $1 - \beta$ if the size of $\x$ is sufficiently large (depending on the values of $\beta$, $\eps$, $\delta$, and $C$).

The basic idea is to apply bounded normalized variance to a private histogram. In particular, the domain of $P$ is partitioned into contiguous bins $B$ of a fixed width, and each bin counts the number of samples from $\x$ that reside in the corresponding subset of the domain. Then, random truncated Laplace noise is added to the count of every bin to ensure privacy. In essence, every bin keeps a noisy count of the number of points $x \in \x$ which land in the bin.

To demonstrate the benefit of this private histogram, suppose that there are two bins $B_1$ and $B_2$ in the domain of $P$, each of which has a sufficiently high noisy count. In particular, if the counts are high enough, then each of bins $B_1$ and $B_2$ must contain at least one $x \in \x$ (i.e., the large counts cannot be created entirely by the truncated Laplace noise). Moreover, any point in the domain between $B_1$ and $B_2$ \emph{must} be an interior point of $\x$. The convenience of this observation is that, even though the exact locations of $\x$ in each of $B_1$ and $B_2$ are unknown (and even the exact number of points in $B_1$ and $B_2$ is only known up to truncated Laplace noise), it is still possible to return an interior point.

There are two possible failure modes that the algorithm might incur. The first is that the points in $\x$ are so spread out that no bin contains very many samples. The second is that the samples in $\x$ are so tightly concentrated, that \emph{only one} bin contains a large number of samples. In order for the algorithm to succeed, we need to ensure that \emph{at least two bins} contain a significant number of samples. 

The key insight is that, if $P$ is $C$-bounded, and if the bin width is chosen in the right way, then the algorithm is guaranteed to succeed with probability at least $1 - \beta$. It turns out that, in order to choose the appropriate bin width, one must first compute an estimate for the first central absolute moment of $P$---this is performed by a subroutine \est which, as we shall discuss in Section~\ref{sec:est}, \emph{also} exploits the $C$-boundedness of $P$. By using a bin width that is slightly smaller than the first-moment estimate produced by \est, we are able to argue that at least two bins will have high noisy counts, and hence that the algorithm will succeed. 

We present the guarantees of Algorithm~\ref{alg:ip-alg} in the following Theorem:

\begin{thm} Suppose we are given four parameters 
$\epsilon > 0$, $\delta \in (0, 1)$, $\beta \in (0,1)$, and $C > 0$. Algorithm~\ref{alg:ip-alg} is $(\epsilon, \delta)$-differentially private. Furthermore, if $P$ has $C$-bounded normalized variance for some $C > 2$, and  $\x \leftarrow P^n$ contains $n$ data points where 
\begin{equation}
n > k_0 C^3\sqrt{\log C} \cdot \left(\epsilon^{-1}\ln \delta^{-1} + \ln \beta^{-1} \right)
\label{eq:nhuge}
\end{equation}
for some sufficiently large positive constant $k_0$, then Algorithm~\ref{alg:ip-alg} returns an interior point of $\x$ with probability at least $1 - \beta$.
\label{thm:ip}
\end{thm} 

\begin{algorithm}
  \caption{Interior Point Algorithm}
  \label{alg:ip-alg}
  \DontPrintSemicolon
  \SetKwFunction{FMain}{interior-point-main}
  \SetKwFunction{Fg}{estimate-first-moment}
  \SetKwFunction{Ff}{find-interior-point}
  \SetKwProg{Fn}{Function}{:}{}
  
  \Fn{\Fg{$\x$; $\eps$,$\delta$,$C$}}{
        Set $k' = 3000$ and $n = |\x|$.\;
        Set $\mathbf{q} = (q_1,\ldots,q_{n/2})$, where $q_{i} = |x_{2i} - x_{2i-1}|$.\;
        For all $\ell \in \mathbb{Z}$, set $c_{\ell}(\q) = \#\{q_i \mid q_i \in (2^{\ell},2^{\ell+1}]\}$ and $\hat{c}_{\ell}(\q) = c_{\ell}(\q) + Z_{\ell}$, where each $Z_{\ell} \leftarrow \Lap(\frac{8}{\eps},\frac{16\ln(16/\delta)}{\eps})$ independently.\;
        Set $S = \left\{\ell \mid \hat{c}_{\ell}(\q) \ge 3n/(8k'C\log C)\right\}$.\;
        \uIf{$|S| \ge 1$}{
            \KwRet $\underset{\ell \in S}{\max}\,2^{\ell+1}$.\;
          }
        \Else{
            \KwRet $\bot$.\;
          }
  }
  \Fn{\Ff{$\x$; $\eps$,$\delta$,$C$,$\hat{m}$}}{
        Set $k' = 3000$, $k = 4096k'$, and $n = |\x|$.\;
        For all $\ell \in \mathbb{Z}$, set $B_{\ell} = [\ell\hat{m}/\up,(\ell+1)\hat{m}/\up)$.\;
        For all $\ell \in \mathbb{Z}$, set $c_{\ell}(\x) = \#\{x_i \mid x_i \in B_{\ell} \}$ and $\hat{c}_{\ell}(\x) = c_{\ell}(\x) + Z_{\ell}$, where each $Z_{\ell} \leftarrow  \Lap(\frac{8}{\eps},\frac{16\ln(16/\delta)}{\eps})$ independently.\;
        Set $S = \left\{\ell \mid \hat{c}_{\ell}(\x) \ge \Tip\right\}$.\;
        \uIf{$|S| \ge 2$}{
            \KwRet $\frac{1}{2}\left(\underset{\ell \in S}{\min}\,\frac{\ell\hat{m}}{\up}+\underset{\ell \in S}{\max}\,\frac{(\ell+1)\hat{m}}{\up}\right)$.\;
          }
        \Else{
            \KwRet $\bot$.\;
          }
  }
  \;
  \SetKwProg{Pn}{Function}{:}{\KwRet}
  \Fn{\FMain{$\x$; $\eps$,$\delta$,$C$}}{
        $\hat{m} \leftarrow \textrm{estimate-first-moment}(\x;\eps,\delta,C)$.\;
        \uIf{$\hat{m} \ne \bot$}{
            \KwRet $\textrm{find-interior-point}(\x; \eps,\delta,C,\hat{m})$.\;
          }
        \Else{
            \KwRet $\bot$.\;
          }
        
  }
\end{algorithm}

\begin{proof}[Proof (Theorem~\ref{thm:ip})]
    We begin by establishing differential privacy. By Lemma~\ref{lem:private-histogram}, the functions \est and \ip each satisfy $(\eps/2, \delta/2)$-differential privacy. In the full algorithm, the output of \est is used as an input for \ip. It follows by the standard composition lemma (see, e.g., \cite{smithprivacy}) that \main satisfies $(\eps, \delta)$-differential privacy.

     Next we turn our attention to the probability of $\main$ returning an interior point. Let $P$ be a $C$-bounded distribution for some $C > 2$ and let $\mu = \E_{X\leftarrow P}[X]$. Finally, let $\x \leftarrow P^n$ where $n$ satisfies \eqref{eq:nhuge}. Critically, the fact that $n$ satisfies \eqref{eq:nhuge} will allow for us to apply Proposition~\ref{prop:estimate} and Proposition~\ref{prop:findip}.
    
    By Proposition~\ref{prop:estimate}, we have with probability at least $1 - \beta/2$ that \estarg returns a value $\hat{m}$ satisfying
    \begin{equation}
    \E_{X\leftarrow P}[|X-\mu|] \le \hat{m} \le \left(6000 C\sqrt{\log C}\right)\E_{X\leftarrow P}[|X-\mu|].
    \label{eq:cond1main}
    \end{equation}
    Conditioned on \eqref{eq:cond1main}, it follows by Proposition~\ref{prop:findip} that \ip returns an interior point of $\x$ with probability at least $1 - \beta /2$.

    Thus, with probability at least $1 - \beta$, $\text{\main}(\x; \epsilon, \delta, C)$ returns an interior point of $\x$.  This completes the proof.
\end{proof}

    


\subsection{Privately Estimating the First Central Absolute Moment}\label{sec:est}

In this section, we introduce \est, a private algorithm for estimating the first central absolute moment of the data distribution $P$, up to a multiplicative factor of $6000C\sqrt{\log C}$. The first central absolute moment approximately measures how much a random variable deviates from its mean on average. More formally, the first central absolute moment is defined as $\E_{X\leftarrow P}[|X-\mu|]$ where $\mu \coloneqq \E_{X\leftarrow P}[X]$ is the mean of the distribution. 

\paragraph{Privately estimating first central absolute moments \emph{without} estimating $\mu$.} In order to calculate the first central absolute moment of $P$, it would be helpful to have a good approximation of $\mu$. Unfortunately, it is hard to \emph{privately} calculate a good approximation to $\mu$ when the samples are unbounded; any function that averages samples together would have unbounded sensitivity, meaning that an enormous amount of noise would need to be added in order to maintain privacy. 

Instead, we consider another strategy for estimating the first central absolute moment of $P$. For independent $X, X'\leftarrow P$ let $Q$ be a random variable that indicates the difference of $X$ and $X'$: $Q \coloneqq |X - X'|$. The random variable $Q$ is advantageous for directly estimating the first central absolute moment of $P$. This is in part due to the expected value $\E[Q]$ being a good proxy to the first central absolute moment of $P$, as shown by Lemma~\ref{lem:b}. Moreover, we will see that the distribution of $Q$ enables us to privately calculate $\E[Q]$.

\paragraph{Overview of the algorithm.} The algorithm \est estimates $E[Q]$ and uses it as a proxy for the first central absolute moment of $P$. It takes as input $\x \leftarrow P^n$ and extracts samples $\q \leftarrow Q^{n/2}$. It then creates a histogram over the domain of $Q$, consisting of contiguous bins whose sizes are increasing powers of $2$. Each bin maintains a count of the number of samples $q \in \q$ that land in the bin, and truncated Laplace noise is added to each count, to maintain privacy. The algorithm then eliminates all bins with small counts. Finally, the algorithm finds the largest of the remaining bins and outputs a fence post of this bin. Critically, the correctness of this algorithm will again rely heavily on the fact that $P$ is $C$-bounded.

To understand why \est returns a good estimate of the first absolute moment of $P$, it helps to focus on the distribution of values for $Q$. We show that the larger values of $Q$ appear with low probability. In particular, values of $Q$ more than $tC$ times larger than $\E[Q]$ have probability that drops as a function of $1/t^2$, as evidenced by Lemma~\ref{lem:d}. Thus, the really large bins (which are simultaneously the bins very far away from $\E[Q]$) will not have much probability mass in expectation and will be eliminated. At the same time, the bins that live very close to $E[Q]$ will receive a large fraction of the mass in expectation, as evidenced by Lemma~\ref{lem:f}. Since \est returns the fence post of the largest bin of those remaining after elimination, then the algorithm is likely to return a point in the domain of $Q$ close to $\E[Q]$. We give the exact details of the performance of \est in Proposition~\ref{prop:estimate}.

\paragraph{Notation.} Assume the data distribution $P$ is $C$-bounded, and let $\mu = \E_{X\leftarrow P}[X]$ denote the mean of samples from $P$. Our privacy parameters are $\epsilon$ and $\delta$. Our confidence parameter is $\beta$: that is, the algorithm outputs the appropriate answer with probability at least $1-\beta$. 

\begin{prop}
\label{prop:estimate}
Let $\eps > 0$, $\delta \in (0,1)$, and $\beta \in (0,1)$. Let $P$ be a $C$-bounded distribution for  $C > 2$. Let $\x \leftarrow P^{2n}$ be a dataset of $2\,n$ data points from $P$ where $n$ satisfies
$$n \ge k\,C\log (C)\left(\ln \left(2/\beta\right) + \frac{16\ln(16/\delta)}{\eps}\right),$$
for a sufficiently  large constant $k$. Then, there exists a constant $k'$ such that {\tt estimate-first-moment($\x$; $\eps$,$\delta$,$C$)} returns an estimate $\hat{m}$ for which the following guarantee holds with probability at least $1-\beta$:
\begin{equation}
    \E_{X\leftarrow P}[|X-\mu|] \le \hat{m} \le \left(2k'C\sqrt{\log C}\right)\E_{X\leftarrow P}[|X-\mu|]\,.
    \label{eq:prop2}
\end{equation}
\end{prop}

To prove Proposition~\ref{prop:estimate}, we use the following lemma, which states that $\E[Q]$ is within a multiplicative factor of $2$ of the first central absolute moment. The proof of this lemma is deferred to the appendix.

\begin{restatable}{lem}{lemb}
\label{lem:b} Let $X$ and $X'$ be two random variables independently drawn from $P$ with mean $\mu$, and let $Q$ be a random variable that indicates the difference of $X$ and $X'$: $Q \coloneqq |X - 
X'|$. Then, we have
$$\E_{X\leftarrow P}[|X-\mu|] \le \E[Q] \le 2\E_{X\leftarrow P}[|X-\mu|]\,.$$
\end{restatable}

As previously mentioned, the bins in the private histogram that are very far away from $\E[Q]$ are highly likely to be eliminated by \est. We demonstrate this with Lemma~\ref{lem:d}, whose proof is deferred to the appendix.
\begin{restatable}{lem}{lemd}
Let $X$ and $X'$ be two random variables independently drawn form a distribution, namely $P$, and let $Q$ be a random variable that indicates the difference of $X$ and $X'$. Suppose $P$ is $C$-bounded for some $C > 1$. For any $t > 0$, we have
$$\Pr\left[Q - \E[Q] \ge tC\E[Q]\right] \le \frac{4}{t^2C}.$$
\label{lem:d}
\end{restatable}
Consider the interval $\mathcal{I} = \left[\frac{1}{2}\E[Q], k'C\sqrt{\log C}\E[Q]\right]$ which surrounds $\E[Q]$ ($k'$ is a constant). We show that there exists a bin in this range that receives a high count in expectation. The proof of Lemma~\ref{lem:f} is deferred to the appendix.

\begin{restatable}{lem}{lemf}
Let $k' = 3000$, $C > 2$. If $P$ is $C$-bounded, then there is some $\ell$ satisfying $$(2^{\ell},2^{\ell+1}] \subseteq\mathcal{I} \text{ and } \Pr\left[Q \in (2^{\ell},2^{\ell+1}]\right] \ge \frac{1}{k'C\log C}.$$
\label{lem:f}
\end{restatable}
We now give the proof of Proposition~\ref{prop:estimate}

\begin{proof}[Proof (Proposition~\ref{prop:estimate})]
Recall that the algorithm \texttt{estimate-first-moment} first creates points $q_1,\ldots,q_n$ and then uses these points to realize a noisy histogram over intervals $(2^{\ell},2^{\ell+1}]$. It then identifies all intervals with $\hat{c}_{\ell}(\q)$ larger than the threshold $\test$. Of these intervals, it chooses the largest $\ell$ and outputs $\hat{m} = 2^{\ell + 1}$ for this $\ell$. By Lemma~\ref{lem:b}, if this largest interval $(2^{\ell},2^{\ell+1}]$ satisfies $(2^{\ell},2^{\ell+1}] \subseteq \mathcal{I}$, then $\hat{m} = 2^{\ell+1}$ satisfies \eqref{eq:prop2}.

We now turn our attention towards the two ways in which \texttt{estimate-first-moment} can fail to output an estimate $\hat{m}$ satisfying \eqref{eq:prop2}. The first mode of failure occurs if there is no such $\ell \in S$ such that $(2^{\ell},2^{\ell + 1}] \subseteq \mathcal{I}$. In particular, we can define $E_1$ to be the event that for all $\ell$ such that $(2^{\ell},2^{\ell+1}] \subseteq I$, $\hat{c}_{\ell}(\q) < \test$. The second mode of failure occurs if the output $2^{\ell+1}$ is too large; in particular, we can define $E_2$ to be the event that there exists an $\ell$ such that $(2^{\ell},2^{\ell+1}] \subseteq [k'C\sqrt{\log C}\E[Q],\infty)$ and $\hat{c}_{\ell}(\q) > \test$. The following two lemmas bound the probability of these bad events occurring:
\begin{lem}
Let $\eps > 0$, $\delta \in (0,1)$, $\beta \in (0,1)$, and let $C > 2$ be parameters. Let $k' = 3000$ and $k = 8k'$. Let $\x \leftarrow P^{2n}$ be the samples fed into the algorithm \est. If we have both that $P$ is $C$-bounded and that $n \ge kC\log C\left(\ln (2/\beta) + 16\ln(16/\delta) / \eps\right),$ then $\Pr[E_1] \le \beta/2$, where the probability is taken over both the randomness of the samples $\x$ and the truncated Laplace mechanism.
\label{lem:e1}
\end{lem}
\begin{proof}
Let $$\ell^* = \argmax_{\ell \text{ s.t. } (2^{\ell},2^{\ell+1}] \subseteq\mathcal{I}} \hat{c}_{\ell}(\q).$$ Lemma~\ref{lem:f} implies that
\begin{equation}
    \label{eq:en}
    \E_{\q}[\csq] = n \Pr\left[Q \in (2^{\ell^*},2^{\ell^*+1}]\right] \ge \frac{n}{k'C\log C}.
\end{equation}
Thus, 
\begin{align*}
    \Pr_{\q,\Zl}[E_1] &= \Pr\left[\cshatq < \Test\right]\\
    &=\Pr\left[\csq + \Zl < \Test\right] \\
    &\le \Pr\left[\csq < \Test + \frac{16\ln(16/\delta)}{\eps} \right]\\
    &\le \Pr\left[\csq < \Test + \frac{n}{8k'C\log C}\right] \tag{by assumption on $n$} \\
    &= \Pr\left[\csq < \frac{4n}{8k'C\log C}\right]\\
    &\le \Pr\left[\csq < \frac{\E[\csq]}{2}\right] \tag{by \eqref{eq:en}}.
\end{align*}
Let $Q_j$ be the indicator random variable for whether $q_j \in (2^{\ell^*},2^{\ell^*+1}]$ and note that $\csq = \sum_{j = 1}^{n} Q_j$. Thus, by a Chernoff bound,
\begin{align*}
    \Pr\left[\csq < \frac{\E[\csq]}{2}\right] &\le \exp\left(-\frac{\E[\csq]}{8}\right)\\
    &\le \exp\left(-\frac{n}{8k'C\log C}\right)\tag{by \eqref{eq:en}}\\
    &\le \exp\left(-\ln (2/\beta)\right)\tag{by assumption on $n$}\\
    &= \beta/2.
\end{align*}
This completes the proof.
\end{proof}

\begin{lem}
Let $\eps > 0$, $\delta \in (0,1)$, and $\beta \in (0,1)$ be parameters. If $P$ is $C$-bounded for some parameter $C > 2$, and $n \ge kC\log C\left(\ln (2/\beta) + 16\eps^{-1}\ln(16/\delta)\right)$, then $\Pr[E_2] \le \beta/2$, where the probability is taken over both the randomness of the samples $\x$ and the truncated Laplace mechanism.
\label{lem:e2}
\end{lem}
\begin{proof}
Define $\csq$ to be the number of points $q_1,\ldots,q_n$ that land in $[k'C\sqrt{\log C}\E[Q],\infty)$. First note that
\begin{align*}
    \E_{\q}[\csq] &= n\Pr\left[Q > k'C\sqrt{\log C}\E[Q]\right]\\
    &= n\Pr\left[Q - \E[Q]> \left(k'C\sqrt{\log C}\E[Q]-1\right)\E[Q]\right]\\
    &\le n\Pr\left[Q - \E[Q]> \frac{k'C\sqrt{\log C}\E[Q]}{2} \right]\\
    &\le \frac{16\,n}{(k')^2C\log C} \tag{by Lemma~\ref{lem:d}}
\end{align*}
and thus
\begin{equation}
    \label{eq:T}
    \E[\csq] \le \frac{n}{k C\log C}
\end{equation}
for sufficiently large $k$. Define $\cshatq \coloneqq \csq + \Zl$ and note that $\cshatq \le \csq + \frac{16\ln(16/\delta)}{\eps}$. This gives us
\begin{align*}
    \Pr[E_2] &\le \Pr\left[\cshatq > \Test\right] \\
    &\le \Pr\left[\csq > \Test - \frac{16\ln(16/\delta)}{\eps} \right] \\
    &\le \Pr\left[\csq > \Test - \frac{n}{8k'C\log C}\right]  \tag{by assumption on $n$}\\
    &= \Pr\left[\csq > \frac{2n}{8k'C\log C} \right].
\end{align*}
Let $Q_j$ be the indicator random variable for whether $q_j \in [k'C\sqrt{\log C}\E[Q],\infty)$, and note that $\csq = \sum_{j = 1}^{n} Q_j$ and $\E[\csq] \le \frac{n}{kC\log C}$ by \eqref{eq:T}. Thus, by a Chernoff bound,
\begin{align*}
\Pr\left[\csq > \frac{2n}{kC\log C} \right]
    &\le \exp\left(-\Omega(\ln (2/\beta)\right)
    \le \beta/2
\end{align*}
which completes the proof.
\end{proof}

The algorithm \texttt{estimate-first-moment} fails to output the desired estimate $\hat{m}$ if either $E_1$ and/or $E_2$ occur. Lemma~\ref{lem:e1} tells us that $\Pr[E_1] \le \beta / 2$, and Lemma~\ref{lem:e2} tells us that $\Pr[E_2] \le \beta / 2$. Thus, by a union bound,
\[
\Pr[E_1 \cup E_2] \le \Pr[E_1] + \Pr[E_2] \le \beta
\]
which implies the proposition.
\end{proof}

\subsection{Finding an Interior Point, Given a Fixed Bin Width}
In this section, we give a guarantee on the success probability of \ip. Recall that the algorithm instantiates a histogram, counting the number of samples in $\x$ that fall into sets of contiguous bins (where the width of each bin is slightly smaller than the output of \est). The algorithm adds truncated Laplace noise to the count of each bin, ensuring that the histogram is private. Then, the algorithm isolates all bins with large counts. Of all the isolated bins, the algorithm picks two and finally returns a value which falls between the domains of each of the two bins. 

\paragraph{Bounded normalized variance induces multiple full bins.} If the algorithm is able to identify multiple bins that each have samples from $\x$, then the algorithm is guaranteed to succeed. The $C$-boundedness assumption on the data distribution guarantees the existence of at least two such bins with high probability (at least $1 - \beta$) over the samples. 

The basic idea behind the analysis is as follows. If there is a large probability mass concentrated in a single bin (but not in any others), then we would be able to use $C$-boundedness in order to deduce that the true first central absolute moment of $P$ is actually much smaller than our bin size---this would contradict Lemma~\ref{lem:e2}. On the other hand, if $P$'s probability mass is so spread out that \emph{no bin} is expected to contain a large noisy count, then we could use $C$-boundedness in order to argue that $P$ violates Chebyshev's inequality, again leading to a contradiction. Thus we are able to conclude (in Lemma~\ref{lem:hardcodeonesided}) that at least two bins should have large noisy counts (with high probability). The full guarantees provided by \ip are laid out in Proposition~\ref{prop:findip}.

\begin{prop}
Let $\eps > 0$, $\delta \in (0,1)$, $\beta \in (0,1)$, $k' = 3000$, $k = 4096k'$. Let $P$ be a $C$-bounded distribution for $C > 2$ and mean $\mu$, and let $\x \sim P^n$ for
\[n \ge \lb.
\]
If $\,\E_{X\leftarrow P}[|X - \mu|] \le \hat{m} \le \up\E_{X\leftarrow P}[|X-\mu|]$, then \iparg returns an interior point of $\x$ with probability at least $1 - \beta$. 
\label{prop:findip}
\end{prop}
\begin{proof} 
To simplify notation, we let $Z = |X-\mu|$ for $X\leftarrow P$. We start by turning our attention to the set $S$, introduced in \ip, which stores the indices to bins with counts above a desirable threshold. Observe that $S$ branches \ip into two cases: either $|S| \ge 2$ or $|S| < 2$. To analyze these two cases, we must begin by making the following claim about $S$ and consequently $\cx$, the non-noisy count of each bin:
\begin{clm}
    For all $\ell \in S$, we have that $\cx > 0$.
    \label{clm:s}
\end{clm}
\begin{proof}[Proof (Claim~\ref{clm:s})]
    By construction, $\ell \in S$ if $\chat > \Tip$. Since $\chat = \cx + \Zl$, this implies that
    \begin{align*}
        \cx &> \Tip - \Zl \\
        &\ge \Tip - \frac{16\ln(16/\delta)}{\eps} \\
        &> 0. \tag{by assumption on $n$}
    \end{align*}
This completes the proof of the claim.
\end{proof}
We show that, in the case where $|S| \ge 2$, the algorithm will always return an interior point. In this case, the algorithm picks two $\ell_1, \ell_2 \in S$ and outputs a point $p$ in the domain that lies between $B_{\ell_1}$ and $B_{\ell_2}$. By Claim~\ref{clm:s}, we know that $B_{\ell_1}$ and $B_{\ell_2}$ each receive at least one sample each from $\x$, and so $p$ must be an interior point of $\x$.

In the case where $|S| < 2$, the algorithm will always fail to output an interior point (since the algorithm defaults to $\bot$ in this case). Thus, we prove the proposition by showing that $|S| < 2$ with probability at most $\beta$. 

To analyze the probability that $|S| < 2$, we need to look at the distribution $P$. It turns out that, if $P$ is $C$-bounded for some known $C > 1$, we are guaranteed that there exists two disjoint regions, at most a distance $\E[Z]/2$ apart, that each contain support in $P$. In particular, we have the following lemma, whose proof is deferred to the appendix:

\begin{restatable}{lem}{lemoneside}
\label{lem:hardcodeonesided}
Suppose $P$ is $C$-bounded for some known $C > 1$. Let $k_1 \ge 2$. Then
\begin{equation}
    \Pr\left[X \in \left(\mu + \frac{\E[Z]}{2k_1}, \mu + 16C\E[Z]\right)\right] \ge \frac{1}{128C}
    \label{eq:lem1}
\end{equation}
and 
\begin{equation}
    \Pr\left[X \in \left(\mu - 16C\E[Z], \mu - \frac{\E[Z]}{2k_1}\right)\right] \ge \frac{1}{128C}.
    \label{eq:lem2}
\end{equation}
\end{restatable}
This implies that, there exists two disjoint intervals $B_{\ell_1}$ and $B_{\ell_2}$ with support in $P$. If $|S| < 2$, then at least one of these two intervals did not receive any samples from $\x$, and either $\ell_1 \notin S$ or $\ell_2 \notin S$. We begin by lower bounding the expected number of samples in $B_{\ell_1}$, i.e. $\E[\cxo]$. Lemma~\ref{lem:hardcodeonesided} tells us that $$\Pr_{X \leftarrow P}\left[X \in \left(\mu + \frac{\E[Z]}{2k_1}, \mu + 16C\E[Z]\right)\right] \ge \frac{1}{128C}.$$
The size of each interval $\Bl$ is $\hat{m}/\up$, and the size of the interval $\left(\mu + \frac{\E[Z]}{2k_1}, \mu + 16C\E[Z]\right)$ is at most $16C\E[Z]$. Thus, the number of intervals $\Bl$ within $\left(\mu + \frac{\E[Z]}{2k_1}, \mu + 16C\E[Z]\right)$ is at most
\[\frac{16C\E[Z]}{\hat{m}/\up} \le \frac{16C\E[Z]\up}{\E[Z]} \le 16C\up.\] This implies that there exists an $\ell_1$ such that $B_{\ell_1} \subseteq \left(\mu + \frac{\E[Z]}{2k_1}, \mu + 16C\E[Z]\right)$ and $$\E[\cxo] \ge \frac{n}{128C}\cdot\frac{1}{16C\up} = \frac{n}{kC^3\sqrt{\log C}}.$$ Thus, it follows that

\begin{align*}
    \Pr_{\x \leftarrow P^n, \Zl}\left[\hat{c}_{\ell_1}(\x) < \Tip\right] &= \Pr\left[\cxo + \Zl < \Tip\right] \\
    &\le \Pr\left[\cxo < \Tip + \frac{16\ln(16/\delta)}{\eps} \right]\\
    &\le \Pr\left[\cxo < \Tip + \frac{n}{kC^3\sqrt{\log C}}\right] \tag{by the assumption on $n$}\\
    &\le \Pr\left[\cxo < \frac{\E[\cxo]}{2}\right] \\
    &\le \exp\left(-\frac{\E[\cxo]}{8}\right) \tag{by a Chernoff bound}\\
    &\le \exp\left(-\frac{\ln (2/\beta)}{8}\right) \tag{by the assumption on $n$}\\
    &\le \frac{\beta}{2}.
\end{align*}

By symmetry, we can also show that there exists an $\ell_2$ satisfying $B_{\ell_2}\subseteq \left(\mu - 16C\E[Z], \mu - \frac{\E[Z]}{2k_1}\right)$ and \\ $\Pr\left[\cxt < \Tip\right] \le \beta/2$. Putting the pieces together, we have that

\begin{align*}
\Pr[|S| < 2] &= \Pr[B_{\ell_1} \notin S \cup B_{\ell_2} \notin S]\\ 
    &\le \Pr\left[\cxo < \Tip\right] + \Pr\left[\cxt < \Tip\right] \\
    &\le \beta/2 + \beta/2 \\
    &= \beta
\end{align*}
which completes the proof.
\end{proof}

\section{Medians}\label{sec:med}
In this section, we introduce a private algorithm (Algorithm~\ref{alg:med-alg}) for finding an $\alpha$-approximate median of a distribution. We show that, if the middle $2\alpha$-percentile of the data distribution is $C$-bounded, then the algorithm returns an $\alpha$-approximation of the median with probability at least $1 - \beta$.

As a convention in this section, we shall use $P$ to refer to the data distribution from which $\x$ is sampled. We will then use $P_{\alpha}$ to refer to the middle $2\alpha$-percentile of $P$, that is,  $P_{\alpha} = P \mid P \in (Q_P(0.5 - \alpha), Q_P(0.5 + \alpha))$. Note that, rather than requiring that $P$ is $C$-bounded, we require that $P_{\alpha}$ is $C$-bounded.


\paragraph{Overview of Algorithm~\ref{alg:med-alg}.} Suppose we had direct sample access to the data distribution $P_{\alpha}$. An interior point of $P_{\alpha}$ is trivially an $\alpha$-approximation to the median of $P$. If $P_{\alpha}$ is $C$-bounded, then by Theorem~\ref{thm:ip}, we could obtain an $\alpha$-approximation to the median. Unfortunately, we cannot assume direct sample access to $P_{\alpha}$ without infinitely-many samples. Thus, Algorithm~\ref{alg:med-alg} instead takes as input the dataset $\x \leftarrow P^n$, isolates samples $\xp \subseteq \x$ which make up \emph{almost} the middle $2\alpha$ fraction of $\x$, and runs Algorithm~\ref{alg:ip-alg} on this smaller dataset $\xp$. While $\xp$ is not sampled i.i.d. from $P_{\alpha}$, we prove that, with high probability, $\xp$ comes from a family of distributions similar to $P_{\alpha}$ that are $C'$-bounded for some $C' = O(C)$. 

To construct $\xp$, we isolate the middle $(2\alpha - 1/k)$-percentile of $\x$, for some parameter $k$ that ends up being a function of $C$ and $\alpha$. The parameter $k$ plays a critical role here, as it guarantees that $\xp$ ends up coming from a distribution that is \emph{contained} in $P_{\alpha}$, rather than from a distribution that \emph{contains} $P_{\alpha}$. As we shall see in the analysis, this distinction allows for us to establish that the distribution $P'$ from which $\xp$ is sampled is $O(C)$-bounded. 

We now introduce Theorem~\ref{thm:med} which gives the formal guarantees of Algorithm~\ref{alg:med-alg}:

\begin{algorithm}[H]
  \caption{Median Algorithm}
  \label{alg:med-alg}
  \DontPrintSemicolon
  \SetKwFunction{FMain}{main}
  \SetKwFunction{Fg}{estimate-first-moment}
  \SetKwFunction{Ff}{find-interior-point}
  \SetKwProg{Fn}{Function}{:}{}
  
  \SetKwProg{Pn}{Function}{:}{\KwRet}
  \Fn{\FMain{$\x$; $\eps$,$\delta$,$\alpha$, $C$}}{
        Let $k = 1024C\alpha^{-1}$.\;
        Let $\xp = \left\{x_i \mid x_i \in \left(Q_{\x}\left(0.5 - \alpha + \frac{1}{2k}\right),Q_{\x}\left(0.5+\alpha - \frac{1}{2k}\right)\right)\right\}$.\;
        \KwRet $\text{\main}(\xp; \eps, \delta, 64C)$. \;
  }
\end{algorithm}

\begin{thm}
Let $\beta, \eps, \delta \in (0,1)$, and $\alpha \in (0,0.25)$. Suppose $P$ is a data distribution such that the conditional distribution on the middle $2\alpha$-percentile of $P$ has  $C$-bounded normalized variance for some $C > 2$. If $\x$ contains $n$ datapoints where $$n \ge k_0\max\left(\frac{C^3\sqrt{\log C}(\eps^{-1}\ln \delta^{-1} + \ln \beta^{-1})}{\alpha},\frac{C^2\ln \beta^{-1}}{\alpha^2}\right)$$ for a sufficiently large positive constant $k_0$, then Algorithm~\ref{alg:med-alg} returns an $\alpha$-approximation of the median with probability at least $1 - \beta$. In addition, for any $\epsilon, \delta \in (0, 1)$ and $C > 0$, we have that Algorithm~\ref{alg:med-alg} is $(\epsilon, \delta)$-differentially private.

\label{thm:med}
\end{thm}

\begin{proof}
    Suppose $\x$ and $\x'$ differ only in one data point. Note that to obtain $\xp$, we sort the elements in $\x$ and take all the elements that have ranks between $\lfloor n \cdot (0.5 - \alpha + \frac{1}{2k})\rfloor$ and $\lfloor n \cdot (0.5+\alpha - \frac{1}{2k})\rfloor$. It is straightforward to show that if we change one data point in $\x$, at most one data point in $\xp$ will be changed. Moreover, previously in Theorem~\ref{thm:ip}, we have shown that the procedure for finding the interior point is $(\epsilon, \delta)$-differentially private. Hence, Algorithm~\ref{alg:med-alg} is $(\epsilon, \delta)$-differentially private.

    By Lemma~\ref{lem:sampledistribution}, we have that with probability at least $1 - \beta/2$ that 
    \begin{equation}
        Q_\x\left(0.5-\alpha+\frac{1}{2k}\right) \in \left(Q_P(0.5-\alpha),Q_P\left(0.5-\alpha+\frac{1}{k}\right)\right)
        \label{eq:Q1}
    \end{equation}
    and 
     \begin{equation}
        Q_\x\left(0.5+\alpha-\frac{1}{2k}\right) \in \left(Q_P\left(0.5+\alpha-\frac{1}{k}\right),Q_P(0.5+\alpha)\right)
        \label{eq:Q2}
    \end{equation} 
    For the rest of the proof, we condition on some arbitrary fixed outcome for the values of $Q_{\x}(0.5-\alpha+\frac{1}{2k})$ and $Q_{\x}(0.5+\alpha-\frac{1}{2k})$ such that \eqref{eq:Q1} and \eqref{eq:Q2} are satisfied.
    
    Note that, once we condition on the outcomes of $Q_\x(0.5-\alpha+\frac{1}{2k})$ and $Q_\x(0.5+\alpha-\frac{1}{2k})$, then $\xp$ consists of i.i.d. samples from the distribution $P_m \coloneqq P \mid P \in (Q_x(0.5-\alpha+\frac{1}{2k}),Q_x(0.5+\alpha-\frac{1}{2k}))$. By \eqref{eq:Q1} and \eqref{eq:Q2}, this distribution $P_m$ can be expressed as
    $$P_m = P|P\in \left(Q_P\left(0.5-\alpha+\frac{1}{k_1}\right),Q_P\left(0.5+\alpha-\frac{1}{k_2}\right)\right)$$
    for some $k_1, k_2 \ge k$. 
    
    Finally, since $k = 2048C/(2\alpha)$, we can apply Lemma~\ref{lem:bignice2}, which says that $P_m$ is $64C$-bounded. Also note that, by assumption on $n$, $$|\xp| = \left(2\alpha - \frac{1}{k}\right)n >k_0 C^3\sqrt{\log C} \cdot \left( \eps^{-1}\ln \delta^{-1} + \ln \beta^{-1} \right).$$ Thus, we can apply Theorem~\ref{thm:ip}, which says that, with probability at least $1-\beta/2$, $\main(\xp,\eps,\delta,64C)$ will return an interior point to $\xp$. This, in turn, is an $\alpha$-approximation of the median of $P$ with probability at least $1-\beta$. 
\end{proof}

We remark that, in the following lemmas, as well as in the statement of Algorithm~\ref{alg:med-alg}, we have not made any effort to optimize the constants involved. We have opted to give explicit constants for concreteness, but with more careful bookkeeping, these constants could almost certainly be made much smaller.

Our next lemma establishes that the endpoints of $\xp$ are guaranteed to (1) be contained within the middle $2\alpha$-percentile of $P$; and (2) be very close to the endpoints of that middle $2\alpha$-percentile.
\begin{lem}
Let $\beta \in (0,1)$, $\alpha \in (0,0.25)$, and $k \ge 1$. Let $P$ be a distribution, and let $\x \leftarrow P^n$ for $n \ge 108k^2\log(4/\beta)$. With probability at least $1-\beta$, $$Q_{\x}\left(0.5 - \alpha + \frac{1}{2k}\right) \in \left(Q_P(0.5 - \alpha),Q_P\left(0.5 - \alpha + \frac{1}{k}\right)\right)$$ and $$Q_{\x}\left(0.5 + \alpha - \frac{1}{2k}\right) \in \left(Q_P\left(0.5 + \alpha - \frac{1}{k}\right),Q_P(0.5 + \alpha)\right).$$
\label{lem:sampledistribution}
\end{lem}
\begin{proof}
Let $X_1,\ldots, X_n$ be $n$ i.i.d. samples from $P$, and define $Y_i$ for all $i \in \{1,\ldots,n\}$ as
\[
Y_i = 
\begin{cases}
1 & \text{if } X_i < Q_P(0.5 - \alpha)\\
0 & \text{otherwise}.
\end{cases}
\]
Define $Y \coloneqq \sum_{i = 1}^n Y_i$. Note that 
$\E[Y] = \sum_{i = 1}^n \Pr[X_i < Q_P(0.5 - \alpha)] = (0.5 - \alpha)n$. Thus,

\begin{align*}
    \Pr\left[Q_{\x}\left(0.5 - \alpha + \frac{1}{2k}\right) \le Q_P(0.5 - \alpha)\right]
    &=\Pr\left[Y \ge \left(0.5 - \alpha + \frac{1}{2k}\right)n\right]\\
    &=\Pr\left[Y \ge (0.5 - \alpha)n\left(1 + \frac{1}{2k(0.5 - \alpha)}\right)\right]\\
    &=\Pr\left[Y \ge \E[Y]\left(1 + \frac{1}{2k(0.5 - \alpha)}\right)\right]\\
    & \le \Pr\left[Y \ge \E[Y]\left(1 + \frac{1}{k}\right)\right]\\
    & \le \exp\left(\frac{-\E[Y]}{3k^2}\right) \tag{by a Chernoff bound}\\
    & = \exp\left(\frac{-(0.5 - \alpha) n}{3k^2}\right) \\
    & \le \exp\left(\frac{-n}{12 k^2}\right) \\
    &\le \beta / 4 \tag{by the assumption on $n$}.
\end{align*}

Likewise, for all $i \in \{1,\ldots,n\}$, let
\[
Z_i = 
\begin{cases}
1 & \text{if } X_i > Q_P\left(0.5 - \alpha + \frac{1}{k}\right)\\
0 & \text{otherwise}.
\end{cases}
\]
Define $Z \coloneqq \sum_{i = 1}^n Z_i$. Note that 
$\E[Z] = \sum_{i = 1}^n \Pr\left[X_i > Q_P\left(0.5 - \alpha + \frac{1}{k}\right)\right] = \left(0.5 - \alpha + \frac{1}{k}\right)n$. Thus,

\begin{align*}
    \Pr\left[Q_{\x}\left(0.5 - \alpha + \frac{1}{2k}\right) \ge Q_P\left(0.5 - \alpha + \frac{1}{k}\right)\right]
    &=\Pr\left[Z \le \left(0.5 - \alpha + \frac{1}{2k}\right)n\right]\\
    &=\Pr\left[Z \le \left(0.5 - \alpha + \frac{1}{k}\right)n\left(1 - \frac{1}{2k(0.5 - \alpha + \frac{1}{k})}\right)\right]\\
    &=\Pr\left[Z \le \E[Z]\left(1 - \frac{1}{2k(0.5 - \alpha + \frac{1}{k})}\right)\right]\\
    &\le \Pr\left[Z \le \E[Z]\left(1 - \frac{1}{3k}\right)\right]\\
    &\le \exp\left(\frac{-\E[Z]}{27k^2}\right) \tag{by a Chernoff bound}\\
    &= \exp\left(\frac{-\left(0.5 - \alpha + \frac{1}{k}\right)n}{27k^2}\right)\\
    &\le \exp\left(\frac{n}{108k^2}\right) \\
    &\le \beta / 4 \tag{by the assumption on $n$}.
\end{align*}
Thus, by a union bound, we have that, with probability at most $\beta/2$, $Q_{\x}(0.5 - \alpha + \frac{1}{2k}) \notin (Q_P(0.5 - \alpha),Q_P(0.5 - \alpha + 1/k))$. By symmetry, it follows that $Q_{\x}(0.5 + \alpha - \frac{1}{2k}) \notin (Q_P(0.5 + \alpha - 1/k),Q_P(0.5 + \alpha))$ with probability at most $\beta / 2$. Thus by another union bound, the lemma holds.
\end{proof}

Our next lemma shows that, if we take a $C$-bounded distribution $P_{\alpha}$ and condition on being in its first $1 - 1/k$ percentile for large enough $k$, then the resulting conditional distribution $P_{\alpha}'$ will be $O(C)$-bounded. Note that the complement of this is not true: if the conditional distribution $P_{\alpha}'$ is $C$-bounded, then the larger distribution $P_{\alpha}$ need not be $O(C)$-bounded. This is why it is critical that $\xp$ is constructed in such a way that it is contained within the middle $2\alpha$-percentile of $P$ (rather than containing the middle $2\alpha$-percentile). 
\begin{lem}
    Let $C > 1$, and let $P_{\alpha}$ be a $C$-bounded distribution. Fix any number $k \ge 128C$. Define $P_{\alpha}' \coloneqq P_{\alpha} \mid P_{\alpha} \in [Q_{P_{\alpha}}(0),Q_{P_{\alpha}}(1-1/k)]$ to be $P_{\alpha}$ conditioned on being in the first $1 - 1/k$ percentile. It follows that $P_{\alpha}'$ is $8C$-bounded.
    \label{lem:bignice}
\end{lem}

\begin{proof}
    We introduce notation which will be used throughout the proof. Let $P_{\alpha}''$ be the distribution $P_{\alpha}$ conditioned on being in the last $1/k$ percentile, i.e. $P_{\alpha}'' \coloneqq P_{\alpha} \mid P_{\alpha} \in [Q_{P_{\alpha}}(1-1/k),Q_{P_{\alpha}}(1)]$. Finally, let $X\leftarrow P_{\alpha}$, $X'\leftarrow P_{\alpha}'$, and $X''\leftarrow P_{\alpha}''$; let $\mu \coloneqq \E_X[X]$, $\mu' \coloneqq \E_{X'}[X']$, and $\mu'' \coloneqq \E_{X''}[X'']$.
    
    To show that $P_{\alpha}$ is $8C$-bounded, we must show that $$\frac{\E_{X'}[|X'-\mu'|^2]}{\E_{X'}[|X'-\mu'|]^2} \le 8C.$$ By the $C$-boundedness of $P_{\alpha}$, it suffices to show
    \begin{equation}
    \frac{\E_{X'}[|X'-\mu'|^2]}{\E_{X'}[|X'-\mu|]^2} \le 8\cdot \frac{\E_X[|X-\mu|^2]}{\E_X[|X-\mu|]^2}.
    \label{eq:something}
    \end{equation}
    We break the proof into two pieces by showing
    \begin{equation}
        \E_{X'}[|X'-\mu'|^2] \le 2\E_{X}[|X-\mu|^2]
        \label{eq:something1}
    \end{equation}
    and 
    \begin{equation}
        \E_{X'}[|X'-\mu'|]^2 \ge \frac{1}{4}\E_{X}[|X-\mu|]^2.
        \label{eq:something2}
    \end{equation}
    To prove \eqref{eq:something1}, note that $\E[|X'-\mu'|^2] \le \E[|X'-\mu|^2]$ since $\mu'$ minimizes the expectation. It follows that
    \begin{align*}
        \E_X\left[|X-\mu|^2\cdot\mathbb{I}_{X < Q_{P_{\alpha}}(1-1/k)}\right] &= \E_X\left[|X-\mu|^2 \mid X < Q_{P_{\alpha}}(1-1/k)\right]\cdot\Pr_X\left[X < Q_{P_{\alpha}}(1-1/k)\right]\\
        &=\E_{X'}\left[|X'-\mu|^2\right](1-1/k)
    \end{align*}
    which rearranges to
    \begin{align*}
        \E_{X'}[|X'-\mu|^2] &= \frac{\E_X\left[|X-\mu|^2\cdot\mathbb{I}_{X < Q_{P_{\alpha}}(1-1/k)}\right]}{1-1/k}\\
        &\le 2\E_X\left[|X-\mu|^2\cdot\mathbb{I}_{X < Q_{P_{\alpha}}(1-1/k)}\right] \tag{by assumption on $k$}\\
        &\le 2\E_X[|X-\mu|^2],
    \end{align*}
    and so indeed \eqref{eq:something1} is true. To prove \eqref{eq:something2}, we begin by expanding $\E_X[|X-\mu|]$ in the following way:
    \begin{align*}
        \E_X[|X-\mu|] &= \Pr_X\left[X < Q_{P_{\alpha}}(1-1/k)\right]\cdot\E_X\left[|X-\mu|\mid X < Q_{P_{\alpha}}(1-1/k) \right] \\
        &\phantom{ooooooo}+ \Pr_X\left[X \ge Q_{P_{\alpha}}(1-1/k)\right]\cdot\E_X\left[|X-\mu|\mid X \ge Q_{P_{\alpha}}(1-1/k)\right]\\
        &= \frac{k-1}{k}\cdot\E_{X'}[|X'-\mu|] + \frac{1}{k}\cdot\E_{X''}[|X''-\mu|]\\
        &\le \frac{k-1}{k}\cdot\E_{X'}[|X'-\mu'|] + \frac{k-1}{k}\cdot|\mu'-\mu| + \frac{1}{k}\cdot\E_{X''}[|X''-\mu|]
    \end{align*}
    which rearranges to
    \begin{equation}
    \E_{X'}[|X'-\mu'|] \ge \frac{k}{k-1}\left(\E_X[|X-\mu|] - \frac{k-1}{k}|\mu'-\mu| - \frac{1}{k}\E_{X''}[|X''-\mu|]\right). 
    \label{eq:P6}
    \end{equation}
    To lowerbound $\E_{X'}\left[|X'-\mu'|\right]$ as in \eqref{eq:something2}, we seek to upperbound $|\mu'-\mu|$ and $\E[|X''-\mu|]$ in terms of $\E_X[|X-\mu|]$. As an intermediate step, we can express $|\mu'-\mu|$ and $\E[|X''-\mu|]$ in terms of $(\mu''-\mu)$ and then upperbound $(\mu''-\mu)$ in terms of $\E_X[|X-\mu|]$. It is not difficult to show that $|\mu - \mu'| = \frac{1}{k-1}(\mu'' - \mu)$ and $\E[|X''-\mu|] = \mu''-\mu$ using both Chebyshev's and the $C$-boundedness of $P_{\alpha}$ (see Claim~\ref{clm:211} and Claim~\ref{clm:qgemu} in Appendix~\ref{sec:app-med} for the full details). 

    We upperbound $\mu''-\mu$ by the following claim, which we prove in Appendix~\ref{sec:app-med}:
    \begin{restatable}{clm}{clmc}
        If $k \ge 128C$, then $\mu'' - \mu \le 3\sqrt{Ck}\E_X[|X-\mu|]$.
        \label{clm:212}
    \end{restatable}
    Putting the pieces together, we see that
    \begin{align*}
        \E_{X'}[|X'-\mu'|] &\ge \frac{k}{k-1}\left(\E_X[|X-\mu|] - \frac{k-1}{k}|\mu'-\mu| - \frac{1}{k}\E_{X''}[|X''-\mu|]\right)\\
        &= \frac{k}{k-1}\left(\E_X[|X-\mu|] - \frac{1}{k}(\mu''-\mu) - \frac{1}{k}(\mu''-\mu)\right)\\
        &\ge \frac{k}{k-1}\left(1- \frac{2}{k}\cdot 3\sqrt{Ck}\right)\E_X[|X-\mu|] \tag{by Claim~\ref{clm:212}}\\
        &\ge \frac{1}{2}\E_X[|X-\mu|] \tag{by assumption on $k$}.
    \end{align*}
    This implies $\E_{X'}[|X'-\mu'|]^2 \ge \frac{1}{4}\E_X[|X-\mu|]^2$, proving \eqref{eq:something2} and thus completing the proof of the lemma.
\end{proof}

Applying Lemma~\ref{lem:bignice} twice, one arrives at the two-sided version of it that we use in the proof of the theorem:

\begin{lem}
    Let $C > 1$, and let $P_{\alpha}$ be a $C$-bounded distribution. Define $P_m \coloneqq P_{\alpha} \mid P_{\alpha} \in [Q_{P_{\alpha}}(1/k_1),Q_{P_{\alpha}}(1-1/k_2)]$ for some $k_1, k_2$. If $k_1, k_2 \ge 2048C$, then $P_m$ is $64C$-bounded.
    \label{lem:bignice2}
\end{lem}
\begin{proof}
   Define $P_{\alpha}' = P_{\alpha} \mid P_{\alpha} \in [Q_{P_{\alpha}}(0),Q_{P_{\alpha}}(1-1/k_2)]$. As $P_{\alpha}$ is $C$-bounded and $k_2 \ge 128C$, we have by Lemma~\ref{lem:bignice} that $P_{\alpha}'$ is $8C$-bounded. 
   Next note that $P_m = P_{\alpha} \mid P_{\alpha} \in  [Q_{P_{\alpha}}(1/k_1),Q_{P_{\alpha}}(1-1/k_2)] = P_{\alpha}' \mid P_{\alpha}' \in \left[ Q_{P_{\alpha}}\left( \frac{1}{k_1 (1 - 1/k_2)}\right), Q_{P_{\alpha}}(1) \right]$. 
   As $P_{\alpha}'$ is $C'$-bounded and $k_1 (1 - 1/k_2) \ge k_1/2 \ge 128C'$, we have by Lemma~\ref{lem:bignice} that $P_m$ is $8C'$-bounded and thus $P_m$ is $64C$-bounded.
\end{proof}

\bibliographystyle{alpha}
\bibliography{bib}

\appendix
\section{Appendix}

\subsection{Technical proofs for \texttt{estimate-first-moment}}
In this section, let $P$ be a data distribution. Let $\mu = \E[P]$. Let $Z = |P - \mu|$. Let $Q = |X - X'|$ for independent $X, X' \leftarrow P$.

\lemb*
\begin{proof}
The upper bound follows from the triangle inequality: 

\begin{align*}
    \E[Q] &= \E[|X - X'|]\\
    &= \E[|X - \mu - X' + \mu|]\\
    &= \E[|(X - \mu) - (X' - \mu)|] \\
    &\le \E[|X - \mu|] + \E[|X' - \mu|]\\
    &= 2\E[|X - \mu|]
    \,.
\end{align*}

For the lower bound, we begin by expanding $\E[Q]$ and the total law of probability:

\begin{align*}
    \E[Q] &=\E_{X,X'}\left[|X - X'|\right] 
    \\&= 
    \Pr\left[X > \mu \right] \cdot \E_{X,X'}\left[|X - X'| \mid X > \mu\right] + \Pr\left[X \leq \mu \right] \cdot \E_{X,X'}\left[|X - X'| \mid X \leq \mu\right]
    \\&\geq 
    \Pr\left[X > \mu \right] \cdot \E_{X,X'}\left[X - X' \mid X > \mu\right] + \Pr\left[X \leq \mu \right] \cdot \E_{X,X'}\left[X' - X \mid X \leq \mu\right]
    \\&= 
    \Pr\left[X > \mu \right] \cdot \E_{X}\left[X - \mu \mid X > \mu\right] + \Pr\left[X \leq \mu \right] \cdot \E_{X}\left[\mu - X \mid X \leq \mu\right]
    \\&= 
    \Pr\left[X > \mu \right] \cdot \E_{X}\left[|X - \mu| \mid X > \mu\right] + \Pr\left[X \leq \mu \right] \cdot \E_{X}\left[|\mu - X| \mid X \leq \mu\right]
    \\& = 
    \E\left[|X - \mu|\right]
    \,.
\end{align*}
\end{proof}

We briefly introduce a claim, useful for the proof of Lemma~\ref{lem:d}.
\begin{clm} Let $X$ and $X'$ be two random variables  independently drawn from $P$, and let $Q$ be the absolute difference of $X$ and $X'$. Then, we have:
$\E[Q^2] = 2\Var(X)$.
\label{clm:a}
\end{clm}

\begin{proof}
\begin{align*}
    \E[Q^2] &= \E\left[|X - X'|^2\right]\\
    &= \E\left[((X - \mu) - (X' - \mu))^2\right]\\
    &= \E\left[(X - \mu)^2 + (X' - \mu)^2 - 2(X - \mu)(X' - \mu) \right]\\
    &= \E\left[(X - \mu)^2\right] + \E\left[(X' - \mu)^2\right] - 2\E\left[(X - \mu)(X' - \mu)\right]\\
    &= 2\E[(X - \mu)^2] - 2\E[X - \mu]\E[X' - \mu] \tag{by independence of $X$ and $X'$}\\
    &= 2\E[(X - \mu)^2]  \\
    &= 2\Var(X)\,.
\end{align*}
\end{proof}

\lemd*
\begin{proof}
We first focus on the variance of $Q$.
\begin{align*}
    \Var(Q) &= \E[Q^2] - \E[Q]^2 \\
    &= 2\,\Var(X) - \E[Q]^2 \tag{by Claim~\ref{clm:a}}\\
    &= 2\,\E[|X-\mu|^2] - \E[Q]^2\\
    &\le 2\,C\,\E[|X-\mu|]^2 - \E[Q]^2 \tag{by assumption}\\
    &\le 4\,C\,\E[Q]^2 - \E[Q]^2 \tag{by Lemma~\ref{lem:b}}\\
    &\le (4\,C-1)\E[Q]^2 \\
    &\le 4\,C\,\E[Q]^2\,.
\end{align*}
Using the above inequality and Chebyshev's inequality, we obtain
\begin{align*}
    \Pr\left[Q - \E[Q] \ge t\,C\,\E[Q]\right]
  & \leq \frac{\Var(Q)}{t^2\, C^2\,\E[Q]^2} \leq \frac{4\,C\,\E[Q]^2}{t^2\, C^2\,\E[Q]^2} = \frac{4}{t^2\,C}\,.
\end{align*}
\end{proof}

We now introduce the following claim, useful for the proof of Lemma~\ref{lem:f}.
\begin{clm}
Suppose $P$ is $C$-bounded for some $C > 1$. Let $\mathcal{I'}$ be the following interval $\left[\frac{1}{2}\E[Q], 128\,C\,\E[Q]\right]$. Then, we have the following bound for the contribution of the domains element in $\mathcal{I'}$ in the expected value of $Q$: $$\E\left[Q \cdot \mathbb{I}_{Q \in \mathcal{I'}}\right] \ge \frac{\E[Q]}{4}\,.$$
\label{clm:e}
\end{clm}

\begin{proof}
One can write a random variable $Q$ via three indicator variables depending where in the domain it belongs:
$$Q = Q \cdot \left(\mathbb{I}_{Q < \E[Q]/2} + \mathbb{I}_{Q \in \mathcal{I'}} + \mathbb{I}_{Q > 128\,C\,\E[Q]} \right)\,.$$
Given this identity, we expand $\E[Q]$ as follows:
\begin{align*}
    \E[Q] &= \underbrace{\E\left[Q \cdot \mathbb{I}_{Q < \E[Q]/2}\right]}_{T_1 \coloneqq} 
    + \E\left[Q \cdot \mathbb{I}_{Q \in \mathcal{I'}}\right]  
    + \underbrace{\E\left[Q \cdot \mathbb{I}_{Q > 128\,C\,\E[Q]}\right]}_{T_2 \coloneqq}
    \,.
\end{align*}
We call the first and the third terms above $T_1$ and $T_2$ respectively. Now, if we rearranges the above identity, we get:
\begin{equation}
    \E\left[Q \cdot \mathbb{I}_{Q \in \mathcal{I'}}\right] = \E[Q] - T_1 - T_2.
    \label{eq:beta}
\end{equation}
Note that our goal is to find a lower bound for $\E\left[Q \cdot \mathbb{I}_{Q \in \mathcal{I'}}\right] $. Hence, we proceed by upper-bounding $T_1$ and $T_2$. It is not hard to see that $T_1$ is bounded by $\E[Q]/2$, since $Q\cdot \mathbb{I}_{Q < \E[Q]/2}$ is never larger than $\E[Q]/2$:
\begin{equation}
    T_1 = \E\left[Q \cdot \mathbb{I}_{Q < \E[Q]/2}\right] 
    \le \frac{\E[Q]}{2}.
    \label{eq:t1}
\end{equation}
Next, we focus on bounding $T_2$ from above. We expand this term via the integral identity for non-negative random variables to get:
\begin{align*}
    T_2 
    & = \E\left[Q \cdot \mathbb{I}_{Q > 128\,C\,\E[Q]}\right] = \int_{0}^{\infty} \Pr[ Q \cdot \mathbb{I}_{Q > 128\,C\,\E[Q]} > q]\,\mathrm{d}q
    \\ & = 
    \int_{q=0}^{128\,C\,\E[Q]} \Pr[Q \cdot \mathbb{I}_{Q > 128\,C\,\E[Q]} > q]\,\mathrm{d}q
    + \int_{q=128\,C\,\E[Q]}^{\infty} \Pr[Q > q]\,\mathrm{d}q
    \,.
\end{align*}
Observe that  $Q \cdot \mathbb{I}_{Q > 128\,C\,\E[Q]}$ is zero whenever $Q \le 128\,C\,\E[Q]$, and it is larger than $128\,C\,\E[Q]$ otherwise. Hence, for any $q \leq 128\,C\,\E[Q]$, the probability of $Q \cdot \mathbb{I}_{Q > 128\,C\,\E[Q]}$ being larger than $q$ is exactly the probability of $Q$ being larger than $128\,C\,\E[Q]$. Thus, we obtain:
\begin{align*}
    T_2 
    & = 128\,C\,\E[Q]\cdot\Pr\left[Q > 128\,C\,\E[Q]\right]
    + \int_{q=128\,C\,\E[Q]}^{\infty} \Pr[Q \ge q]\,\mathrm{d}q
    \\ & =
    128\,C\,\E[Q]\cdot\Pr\left[Q > 128\,C\,\E[Q]\right]
    + C\E[Q] \cdot \int_{t = 128}^{\infty} \Pr\left[Q \ge tC\E[Q]\right]\,\mathrm{d}t
    \,,
\end{align*}
where the last line is due to the change of variable in the integral using $q = t\,C\,\E[Q]$.
Now, we use Lemma~\ref{lem:d} to bound the probability terms above:
\begin{equation}
   \label{eq:t2}
    \begin{split}
    T_2 & \leq
    128\,C\,\E[Q]\cdot \frac{4}{(128)^2C} 
    + C\E[Q]\cdot \int_{t = 128}^{\infty} \frac{4}{t^2\,C} \,\mathrm{d}t 
    \\
    &= \frac{4\E[Q]}{128} + \frac{4\E[Q]}{128} = \frac{\E[Q]}{16}
    \,.
    \end{split}
\end{equation}

By \eqref{eq:beta}, \eqref{eq:t1}, and \eqref{eq:t2}, we have that:
\begin{align*}
    \E\left[Q \cdot \mathbb{I}_{Q \in \mathcal{I'}}\right] &= \E[Q] - T_1 - T_2
    \ge \E[Q] - \frac{\E[Q]}{2} - \frac{\E[Q]}{16} 
    \ge \frac{\E[Q]}{4}\,.
\end{align*}
Hence, the proof of the lemma is complete. 
\end{proof}

We can now prove Lemma~\ref{lem:f}.
\lemf*
\begin{proof}
Let $k' = 3000$, $C > 2$, and $\mathcal{I'} = \left[\frac{1}{2}\E[Q], 128\,C\,\E[Q]\right]$. By Claim~\ref{clm:e}, we have that 
\begin{equation}
\label{eq:p1}
\E[Q \cdot \mathbb{I}_{Q \in \mathcal{I'}}] \ge \E[Q]/4. 
\end{equation}
On the other hand,
\begin{equation}
\label{eq:p2}
\E[Q \cdot\mathbb{I}_{Q \in \mathcal{I'}}] = \int_{q\in \mathcal{I'}} q \cdot\Pr[Q=q]\,\mathrm{d}q \le 128\,C\,\E[Q]\cdot \Pr[Q \in \mathcal{I'}].
\end{equation}
Together \eqref{eq:p1} and \eqref{eq:p2} give
$$\frac{\E[Q]}{4} \le 128\,C\,\E[Q]\cdot \Pr[Q \in \mathcal{I}'],$$
which implies that
$$\Pr[Q \in \mathcal{I}'] \ge \frac{1}{4\cdot 128C}.$$
Let $\ell^* \in \mathbb{Z}$ be such that $(2^{\ell^*},2^{\ell^* + 1}] \subseteq \mathcal{I}'$ is the set in $\mathcal{I}'$ with the most probability mass. It follows that
\begin{align*}
\Pr[Q \in (2^{\ell^*},2^{\ell^* + 1}]] &\ge \frac{1}{4\cdot 128C} \cdot \frac{1}{|\{(2^{\ell},2^{\ell+1}] \mid (2^{\ell},2^{\ell+1}] \subseteq \mathcal{I'}\}|} \\
&\ge \frac{1}{4\cdot 128C} \cdot \frac{1}{\log (128\,C\,\E[Q]) - \log (\E[Q]/2) + 1}\\
&= \frac{1}{4\cdot 128C} \cdot \frac{1}{\log (2\cdot 128C)+1}\\
&\ge \frac{1}{k'C\log C} \tag{by assumptions on $k'$, $C$}.
\end{align*}

Finally, let $\ell' \in \mathbb{Z}$ be such that $(2^{\ell'},2^{\ell'+1}] \subseteq \mathcal{I}$ is the set in $\mathcal{I}$ with the most probability mass. Note that $\mathcal{I'} \subseteq \mathcal{I}$ by the assumption on $k$, and thus $$\Pr[Q \in (2^{\ell'},2^{\ell'+1}]] \ge \Pr[Q \in (2^{\ell^*},2^{\ell^*+1}]] \ge \frac{1}{k'C\log C}$$ which completes the proof.
\end{proof}

\subsection{Technical proofs for \texttt{find-interior-point}}

\begin{lem}
For a random variable $X$ and constant $\mu$, let $Z = |X - \mu|$, and suppose that $\E[Z]^2 \le \E[Z^2] \le C \E[Z]^2$ for some known $C > 1$. Let $b = \E[Z]/k_1$ for some $k_1 \ge 1$. It follows that $$\Pr[X \in (\mu - b/2, \mu + b/2)] \le 1 - \frac{1}{16C}.$$
\label{lem:middlebin}
\end{lem}

\begin{proof}
Let $p = \Pr[X \in (\mu - b/2, \mu + b/2)]$, and let $\sigma = \sqrt{\E[Z^2]}$. First note that
\[
    b = \frac{\E[Z]}{k_1} < \E[Z] \le \sigma \le \frac{\sigma}{\sqrt{1-p}}
\]
and thus $$\frac{b}{2} < \frac{\sigma}{\sqrt{1-p}}.$$

Expanding $\E[|X-\mu|]$ gives
\begin{align*}
    \E[|X-\mu|] &= \E[Z] \\
    &= \int_{z=0}^{\infty} \Pr[Z \ge z] dz \\
    &= \int_{z=0}^{b/2} \Pr[Z \ge z]dz + \int_{z=b/2}^{\sigma / \sqrt{1-p}} \Pr[Z \ge z]dz + \int_{z=\sigma/\sqrt{1-p}}^{\infty} \Pr[Z \ge z]dz\\
    &= T_1 + T_2 + T_3.
\end{align*}
We will now bound each of $T_1$, $T_2$, and $T_3$ from above.

We first have the following bound on $T_1$,
\begin{align*}
    T_1 &= \int_{z=0}^{b/2} \Pr[Z \ge z]dz \\
    &\le \int_{z=0}^{b/2} 1 \cdot dz \\
    &= b/2.
\end{align*}

Note that a $1 - p$ fraction of mass is outside of $(\mu - b/2, \mu + b/2)$. We can bound $T_2$ by assuming that this $1 - p$ fraction is as far away from $\mu$ as possible, i.e. this $1 - p$ fraction sits at least $\frac{\sigma}{\sqrt{1 - p}}$ away from $\mu$. Thus,
\begin{align*}
    T_2 &= \int_{z=b/2}^{\sigma / \sqrt{1-p}} \Pr[Z \ge z]dz \\
    &\le \int_{z=b/2}^{\sigma / \sqrt{1-p}} (1 - p)dz \\
    &= (1-p)z \bigg |_{b/2}^{\sigma/ \sqrt{1-p}} \\
    &= (1-p)\left(\frac{\sigma}{\sqrt{1-p}} - \frac{b}{2}\right) \\
    &= \sigma(\sqrt{1-p}) - \frac{b}{2} + \frac{bp}{2}.
\end{align*}

We can bound $T_3$ using Chebyshev's inequality:
\begin{align*}
    T_3 &= \int_{z=\sigma / \sqrt{1-p}}^{\infty} \Pr[Z \ge z] dz\\
    &= \int_{k = 1 / \sqrt{1-p}}^{\infty} \Pr[Z \ge k\sigma] \sigma dk \tag{By a change of variables}\\
    &\le \int_{k = 1 / \sqrt{1-p}}^{\infty} \frac{\sigma}{k^2} dk \tag{By Chebyshev's Inequality}\\
    &= \frac{-\sigma}{k} \bigg |_{1 / \sqrt{1-p}}^{\infty} \\
    &= \sigma \sqrt{1-p}.
\end{align*}

Thus,

\begin{align*}
    \E[|X-\mu|] &\le T_1 + T_2 + T_3 \\
    &\le \frac{b}{2} + \sigma(\sqrt{1-p}) - \frac{b}{2} + \frac{bp}{2} + \sigma(\sqrt{1-p}) \\
    &= \frac{bp}{2} + 2\sigma(\sqrt{1-p}) \\
    &\le \frac{\E[|X-\mu|]}{2k_1}\cdot p + 2\sqrt{C}\cdot \E[|X-\mu|]\sqrt{1-p}.
\end{align*}

Dividing both sides by $\E[|X-\mu|]$, we have 
$$1 \le \frac{1}{2k_1}p + 2\sqrt{C}\sqrt{1-p}.$$
Since $k_1 > 1$ and $p \le 1$, it follows that
$$1 \le \frac{1}{2} + 2\sqrt{C}\sqrt{1-p},$$
which rearranges to 
$$p \le 1 - \frac{1}{16C}.$$
\end{proof}

\begin{lem}
For a random variable $X$ and constant $\mu$, let $Z = |X - \mu|$, and suppose that $\E[Z^2] \le C \E[Z]^2$ for some known $C > 1$. Let $t > 0$. It follows that $$\Pr\left[X \in (\mu - k\sqrt{C}\E[Z], \mu + t\sqrt{C}\E[Z])\right] \ge 1 - \frac{1}{t^2}.$$
\label{lem:cheb}
\end{lem}

\begin{proof}
Let $\sigma = \sqrt{\E[Z^2]} = \textrm{Var}[Z]$. (Note that $\sigma \le \sqrt{C}\E[Z]$.) The lemma immediately follows from Chebyshev's inequality:
\begin{align*}
    \Pr\left[X \in (\mu - t\sqrt{C}\E[Z], \mu + t\sqrt{C}\E[Z])\right] &= \Pr[Z \le t\sqrt{C}\E[Z]] \\
    &\ge \Pr[Z \le t\sigma] \\
    &= 1 - \Pr[Z \ge t\sigma] \\
    &\ge 1 - \frac{1}{t^2}.
\end{align*}
\end{proof}

\begin{lem}
For a random variable $X$ and constant $\mu$, let $Z = |X - \mu|$, and suppose that $\E[Z^2] \le C \E[Z]^2$ for some known $C > 1$. Let $b = \E[Z]/k_1$ for some $k_1 \ge 1$. Let $S = (\mu - 8C\E[Z], \mu + 8C\E[Z]) \setminus (\mu - b/2, \mu + b/2)$. It follows that $$\Pr[X \in S] \ge \frac{1}{32C}.$$
\label{lem:outsidemiddlebin2}
\end{lem}
\begin{proof}
By Lemma~\ref{lem:cheb}, it follows that $$\Pr\left[X \in (\mu - 8C\E[Z], \mu + 8C\E[Z])\right] \ge 1 - \frac{1}{64C}.$$
Likewise, by Lemma~\ref{lem:middlebin}, it follows that
$$\Pr[X \in (\mu - b/2, \mu + b/2)] \le 1 - \frac{1}{16C}$$

Thus, we have
\begin{align*}
    \Pr[X \in S] &\ge \Pr\left[X \in (\mu - 8C\E[Z], \mu + 8C\E[Z])\right] - \Pr[X \in (\mu - b/2, \mu + b/2)] \\
    &\ge 1 - \frac{1}{64C} - \left(1 - \frac{1}{16C}\right)\\
    &\ge \frac{1}{32C}.
\end{align*}

\end{proof}

We introduce the following claim, which is used in an ensuing lemma.
\begin{clm}
For any random variable $X$ satisfying $\E[X] = \mu$, 
$$\E[(X - \mu) \cdot \mathbb{I}_{X \ge \mu}] = \frac{\E[|X - \mu|]}{2}.$$
\label{clm:onesided}
\end{clm}

\begin{proof}
\begin{align*}
    0 &= \E[X-\mu] \\
    &= \int_{x=-\infty}^{\infty} (x-\mu)\Pr[X=x]\,\mathrm{d}x \\
    &= \int_{x=-\infty}^{\mu}(x-\mu)\Pr[X=x]\,\mathrm{d}x + \int_{x=\mu}^{\infty}(x-\mu)\Pr[X=x]\,\mathrm{d}x.
\end{align*}
This implies

\begin{equation}
    -\int_{x=-\infty}^{\mu}(x-\mu)\Pr[X=x]\,\mathrm{d}x = \int_{x=\mu}^{\infty}(x-\mu)\Pr[X=x]\,\mathrm{d}x.
\end{equation}
Likewise,

\begin{align*}
    \E[|X-\mu|] &= \int_{x=-\infty}^{\infty} |x-\mu|\Pr[X=x]\,\mathrm{d}x \\
    &= -\int_{x=-\infty}^{\mu}(x-\mu)\Pr[X=x]\,\mathrm{d}x + \int_{x=\mu}^{\infty}(x-\mu)\Pr[X=x]\,\mathrm{d}x.
\end{align*}
This implies that $$\E[(X - \mu) \cdot \mathbb{I}_{X \ge \mu}] = \int_{x = \mu}^{\infty} (x - \mu) \Pr[X = x]\,\mathrm{d}x = \frac{\E[|X-\mu|]}{2}.$$

\end{proof}
Finally, we hard-code $\mu$ as $\E[X]$, and show that there is a non-negligable fraction of probability mass in both $(\mu - 16C\E[Z], \mu - b/2)$ and $(\mu + b / 2, \mu + 16C\E[Z])$.

\lemoneside*

\begin{proof}
Set $X' = \max(X,  \mu)$. By Claim~\ref{clm:onesided}, if we define $Z' = X' - \mu$, we have $$\E[Z'] = \E[X' - \mu] = \frac{\E[|X - \mu|]}{2} = \frac{\E[Z]}{2}.$$ Furthermore, we note that $\E[(Z')^2] = \E[(X' - \mu)^2] \le \E[(X - \mu)^2] = \E[Z^2].$ Combining these two facts, it follows that

\begin{align*}
    \E[(Z')^2] &\le \E[Z^2] \\
    &\le C\E[Z]^2 \\
    &= 4C \E[Z']^2.
\end{align*}
Setting $C' = 4C$ and $b' = \E[Z'] / (k_1/2)$, we define $S'$ as follows: 
\begin{align*}
    S' &\coloneqq (\mu - 8C'\E[Z'], \mu + 8C'\E[Z']) \setminus (\mu - b'/2, \mu + b'/2)\\
    &= \left(\mu - 8\cdot 4C \cdot \frac{\E[Z]}{2},\mu + 8\cdot 4C \cdot \frac{\E[Z]}{2}\right) \setminus \left(\mu - \frac{\E[Z']}{2(k_1/2)}, \mu + \frac{\E[Z']}{2(k_1/2)}\right) \\
    &= (\mu - 16C \E[Z],\mu + 16C\E[Z] \setminus \left(\mu - \frac{\E[Z']}{k_1}, \mu + \frac{\E[Z']}{k_1}\right) \\
    &= (\mu - 16C \E[Z],\mu + 16C\E[Z] \setminus \left(\mu - \frac{\E[Z]}{2k_1}, \mu + \frac{\E[Z]}{2k_1}\right) \\
    &= (\mu - 16C \E[Z],\mu + 16C\E[Z] \setminus (\mu - b/2, \mu + b/2).
\end{align*}
We can apply Lemma~\ref{lem:outsidemiddlebin2} to $X'$ to deduce that $$\Pr[X' \in S'] \ge \frac{1}{32C'}= \frac{1}{128C}.$$ Since $X' = \max(X,\mu)$, this implies \eqref{eq:lem1}. One can similarly show \eqref{eq:lem2} by symmetry.
\end{proof}

\subsection{Technical proofs for \texttt{median}}\label{sec:app-med}

\begin{restatable}{clm}{clmb}
    Let $k \ge 128C$. It follows that $Q_{P_{\alpha}}(1-1/k) > \mu$.
    \label{clm:qgemu}
\end{restatable}
\begin{proof}
    First, note that, when $k \ge 128C$,
    \begin{align*}
        \Pr\left[X \ge \mu + \frac{\E[|X-\mu|]}{4}\right] &\ge \frac{1}{128C}\tag{by Lemma~\ref{lem:hardcodeonesided}}\\ 
        &\ge \frac{1}{k} \\
        &= \Pr\left[X \ge Q_P\left(1-\frac{1}{k}\right)\right]
    \end{align*}
    which implies that $Q_{P_{\alpha}}(1-1/k) > \mu$. 
\end{proof}

\begin{restatable}{clm}{clma}
    $|\mu - \mu'| = \frac{1}{k-1}(\mu'' - \mu)$.
    \label{clm:211}
\end{restatable}
\begin{proof}
    By the definitions of $\mu$, $\mu'$, and $\mu''$ we have that $$\mu = \left(1-\frac{1}{k}\right)\mu'+\frac{1}{k}\mu''.$$ Rearranging gives
    \begin{align*}
        \mu - \frac{1}{k}\mu &= \left(1-\frac{1}{k}\right)\mu'+\frac{1}{k}\mu'' - \frac{1}{k}\mu\\
        \left(1-\frac{1}{k}\right)(\mu-\mu') &= \frac{1}{k}(\mu''-\mu)\\
        \mu-\mu' &= \frac{1}{k-1}(\mu''-\mu).\\
    \end{align*}
    By Claim~\ref{clm:qgemu} and assumption on $k$, we have that $Q_{P_{\alpha}}(1-1/k) > \mu$ which implies that $\mu''-\mu \ge 0$. Thus, $$|\mu-\mu'| = \frac{1}{k-1}(\mu''-\mu).$$
\end{proof}

\clmc*
\begin{proof}
    By Claim~\ref{clm:qgemu}, $Q_P(1-1/k) > \mu$. Also,
    \begin{align*}
        \Pr\left[X -\mu \ge \sqrt{C2^ik\E[|X-\mu|]^2}\right]&\le \Pr\left[|X -\mu| \ge \sqrt{C2^ik\E[|X-\mu|]^2}\right]\\
        &\le \frac{\E[|X-\mu|^2]}{C2^ik\E[|X-\mu|]^2} \tag{by Chebyshev's inequality}\\
        &\le \frac{C\E[|X-\mu|]^2}{C2^ik\E[|X-\mu|]^2} \tag{by the $C$-bounded property}\\
        &\le \frac{1}{2^ik}.
    \end{align*}
    These equations imply that $$\Pr\left[X-\mu \ge \sqrt{C2^ik\E[|X-\mu|]^2}\right] \le \Pr\left[X - \mu \ge Q_P\left(1-\frac{1}{2^ik}\right) - \mu\right]$$ and thus 
    \begin{equation}
        Q_P\left(1-\frac{1}{2^ik}\right) -\mu \le \sqrt{C2^ik\E[|X-\mu|]^2}.
    \end{equation}

    So note that
    \begin{align*}
        \mu'' - \mu &= \int_{x=Q_P(1-1/k)}^{Q_P(1)} (x-\mu) \Pr[X''=x]\,\mathrm{d}x\\
        &= \sum_{i=0}^{\infty}\int_{x=Q_P\left(1-\frac{1}{2^{i}k}\right)}^{Q_P\left(1-\frac{1}{2^{i+1}k}\right)} (x-\mu)\Pr[X''=x]\,\mathrm{d}x\\
        &\le \sum_{i=0}^{\infty}\int_{x=Q_P\left(1-\frac{1}{2^ik}\right)}^{Q_P\left(1-\frac{1}{2^{i+1}k}\right)} (Q_P\left(1-\frac{1}{2^{i+1}k}\right)-\mu)\Pr[X''=x]\,\mathrm{d}x\\
        &\le \sum_{i=0}^{\infty}\int_{x=Q_P\left(1-\frac{1}{2^ik}\right)}^{Q_P\left(1-\frac{1}{2^{i+1}k}\right)} (Q_P\left(1-\frac{1}{2^{i+1}k}\right)-\mu)k\Pr[X=x]\,\mathrm{d}x\\
        &\le \sum_{i=0}^{\infty}(Q_P\left(1-\frac{1}{2^{i+1}k}\right)-\mu)\cdot\frac{1}{2^{i+1}}\\
        &\le \sum_{i=0}^{\infty}\sqrt{C2^{i+1}k\E[|X-\mu|]^2}\cdot\frac{1}{2^{i+1}}\\
        &= \sqrt{Ck}\E[|X-\mu|]\sum_{i=0}^{\infty} \frac{1}{\sqrt{2^{i+1}}}\\
        &\le 3\sqrt{Ck}\E[|X-\mu|].
    \end{align*}
This completes the proof of the claim.
\end{proof}

\subsection{Truncated Laplace mechanism} \label{sec:TLap}

\lemTLap*

\begin{proof} 	
Our goal is to show that for every two neighboring datasets $\x$ and $\x$ in $\mathcal{X}^n$ and every measurable set $S$:

$$\Pr\left[M(\x) \in S\right] \leq e^\eps \cdot \Pr\left[M(\x') \in S\right] + \delta\,.$$

We use the following notation in this proof: Let $I \coloneqq [f(\x) - \Zmax, f(\x) + \Zmax]$ and $I' \coloneqq [f(\x') - \Zmax, f(\x') + \Zmax]$ indicate the domain of $M(\x)$ and $M(\x')$ respectively. Let $\Psi$ be the normalizing constant in the truncated Laplace distribution, which is: 
	\begin{align*}
		\Psi \coloneqq \int_{z=-\Zmax}^{\Zmax} e^{-\frac{\eps\,|z|}{\Delta}}dz = \frac{2\,\Delta}{\eps} \cdot \left(1-e^{-\frac{\eps\,\Zmax}{\Delta}}\right)\,.
	\end{align*}

Given the above notation, we focus on bounding the probability of $M(\x)$ being in a measurable set $S$. One can partition~$S$ into two sets, $S \cap I'$ or $S \setminus I'$, and write:

\begin{align*}
	\Pr\left[M(\x) \in S\right] 
    & = \Pr\left[M(\x) \in S \cap I'\right]
	+
	\Pr\left[M(\x) \in S \setminus I'\right]
	\\& = \int_{u \in S \cap I'} \Pr\left[M(\x) = u\right] \, du 
	+
	\Pr\left[M(\x) \in S \setminus I'\right]
	\\
	& \leq  
	\int_{u \in S \cap I' } \Pr\left[f(\x) + Z = u\right] \, du 
	+ 
	\Pr\left[M(\x) \not \in I'\right]
	\\ &\leq 
	\int_{u \in S \cap I' } \frac{e^{-\frac{\eps |u-f(\x)|}{\Delta}}}{\Psi} \,du
	+ 
	\Pr\left[M(\x) \not \in I'\right]\,,
\end{align*}
where in the last line we use the definition of the truncated Laplace distribution. Next, we leverage the fact that the sensitivity of $f$ is bounded by $\Delta$ (i.e., $|f(\x) - f(\x')| \leq \Delta$). Thus, we have:

\begin{align*}
	\Pr\left[M(\x) \in S\right]  & \leq 
	\int_{u \in S \cap I' } \frac{e^{-\frac{\eps \left(|u-f(\x')| - |f(\x) - f(\x')|\right)}{\Delta}}}{\Psi} \,du
	+ 
	\Pr\left[M(\x) \not \in I'\right]
	\\ & \leq 
	\int_{u \in S \cap I' } \frac{e^{-\frac{\eps \left(|u-f(\x')| - \Delta\right)}{\Delta}}}{\Psi} \,du
	+ 
	\Pr\left[M(\x) \not \in I'\right]
	\\ & \leq e^\eps \cdot 
	\int_{u \in S \cap I' } \frac{e^{-\frac{\eps \left(|u-f(\x')|\right)}{\Delta}}}{\Psi} \,du
	+ 
	\Pr\left[M(\x) \not \in I'\right]\,.
\end{align*}
In the first term above, $u$ is in $I'$. Hence, by the definition of the truncated Laplace mechanism, we get:
\begin{align*}
	\\ & \leq e^\eps \cdot 
	\int_{u \in S \cap I' } \Pr\left[f(\x') + Z = u\right] \,du
	+ 
	\Pr\left[M(\x) \not \in I'\right]
	\\ & = e^\eps \cdot \Pr\left[M(\x') \in S \cap I'\right]
	+ 
	\Pr\left[M(\x) \not \in I'\right]
	\\ & \leq e^\eps \cdot \Pr\left[M(\x') \in S\right]
	+ 
	\Pr\left[M(\x) \not \in I'\right]
	\,.
\end{align*}

Now, to prove our statement, it suffices to show that the probability of $M(\x)$ landing outside of $I'$ is at most $\delta$. Given the symmetry in the Laplace noise, without loss of generality, assume $f(\x) \geq f(\x')$; Otherwise, one can multiply $f(\x)$ and $f(\x')$ and $Z$'s by $-1$, and the proof remains unchanged. We know that  $M(\x)$ must be in $[f(\x) - \Zmax, f(\x) + \Zmax]$. Now, given that $\x$ and $\x'$ are two neighboring datasets, we know that $f(\x')$ is at least $f(\x) - \Delta$. Therefore, the probability of $M(\x)$ being ouside of $I'$ is bounded by the probability of $Z \in [\Zmax-\Delta, \Zmax]$ where $Z$ is a drawn from $\Lap(\lambda,\Zmax)$. Note that every $Z$ in this interval is positive, so we obtain:
	\begin{align*}
		\Pr\left[M(\x) \not \in I'\right] & \leq \Pr\left[Z \in [\Zmax-\Delta, \Zmax]\right] = \frac{1}{\Psi} \cdot \int_{z = \Zmax - \Delta}^{\Zmax}  e^{-\frac{\eps\,|z|}{\Delta}} dz
		\\
		&= \frac{1}{\Psi} \cdot \int_{z = \Zmax - \Delta}^{\Zmax}  e^{-\frac{\eps\,z}{\Delta}} dz = \frac{\Delta \left(e^{-\frac{\eps\, (\Zmax - \Delta)}{\Delta}} - e^{-\frac{\eps\, \Zmax}{\Delta}}\right)}{\Psi \cdot \eps}
		\\
		& = \frac{e^{-\frac{\eps \Zmax}{\Delta}} \cdot \left(e^\eps - 1\right)}{2 \cdot \left(1- e^{-\frac{\eps \Zmax}{\Delta}}\right)} 
		\leq \frac{e^{-\frac{\eps\Zmax}{\Delta}}}{1-e^{-\frac{\eps\Zmax}{\Delta}}} \leq \delta\,. 
	\end{align*}
Hence, the proof is complete. 
\end{proof}

\end{document}